\documentclass[journal]{IEEEtran}
\usepackage{amsmath}
\usepackage{amsfonts}
\usepackage{amssymb}
\usepackage{bbm}
\usepackage{graphicx}
\usepackage{textcomp}
\usepackage{multirow}
\usepackage[linesnumbered,ruled]{algorithm2e}
\usepackage[switch]{lineno}
\usepackage{amsthm}
\usepackage{amsopn}
\usepackage{url}
\usepackage{epstopdf}
\usepackage{color}
\usepackage{bm}
\usepackage{verbatim}
\usepackage{enumerate}
\usepackage{comment}
\usepackage [short,nocomma]{optidef}
\usepackage{mathtools}
\usepackage{booktabs}
\usepackage[caption=false]{subfig}
\usepackage{array}
\usepackage{stfloats}
\def\BibTeX{{\rm B\kern-.05em{\sc i\kern-.025em b}\kern-.08em
		T\kern-.1667em\lower.7ex\hbox{E}\kern-.125emX}}
\usepackage{booktabs}
\usepackage{pifont}

\newtheorem{prop}{Proposition}
\SetKwInput{KwInput}{Input}
\SetKwInput{KwOutput}{Output}

\usepackage{threeparttable}
\usepackage{array}
\usepackage{multirow}
\usepackage{longtable}
\usepackage{rotating}
\newcommand{\tabincell}[2]{\begin{tabular}{@{}#1@{}}#2\end{tabular}}
\ifCLASSINFOpdf
\else
\fi
\begin{document}
\title{Orchestrating Data Collection and Computation in Green IoT Networks}
\author{Junfei Zhan, Tengjiao He, Kwan-Wu Chin, Benyu Chen, Fei Song
\thanks{
        Author Zhan is with the Jinan University-University of Birmingham Joint Institute, Jinan University, Guangzhou, China, and also with the Department of Electrical and Systems Engineering, University of Pennsylvania, PA, USA (email: zjf2024@seas.upenn.edu).
        Author He and Chen are with the College of Information Science and Technology, Jinan University, Guangzhou, China (email: htj2018@jnu.edu.cn, 202334071005@stu.jnu.edu.cn.
        Author Chin is with the School of Electrical, Computer and Telecommunications Engineering, University of Wollongong, NSW, Australia (email: kwanwu@uow.edu.au).
        Author Song is with the School of Electronic and Information Engineering, Beijing Jiaotong University, Beijing, China  (email: fsong@bjtu.edu.cn).
}
}

\maketitle
\newtheorem{proposition}{Proposition}

\begin{abstract}
Future Internet of things (IoT) networks will host applications that involve data collection and computation tasks on one or more servers.  
To this end, this paper proposes the first mixed integer linear program (MILP) to schedule and embed applications on energy harvesting nodes, where it optimizes (i) the sampling time of devices, (ii) whether to run an application, and (iii) the energy usage of devices, gateways and servers.  
To ensure applications are run often, we adopt the maximum age of service (AoS) metric, and set the MILP's objective to minimize the maximum AoS or min-max AoS of applications.
This paper also proposes two novel solutions: 
(i) a receding horizon control (RHC) based method, and 
(ii) a solution that greedily embeds applications according to their AoS.  
The results show that the min-max AoS of RHC and greedy approach is respectively 1.07x and 1.13x higher than MILP. 
\end{abstract}

\begin{IEEEkeywords}
Renewable energy, Resource allocation, Model Predictive Control, Virtualization, Information Freshness.
\end{IEEEkeywords}
\IEEEpeerreviewmaketitle

\section{Introduction}
In the near future, Internet of things (IoT) networks will have to optimize computing, communication, and storage resources to run artificial intelligence (AI) driven applications~\cite{8658105}. 
In addition, they will employ virtualization, whereby devices and servers run containers or virtual machines so that different users can share an IoT network \cite{10634801}.
Apart from that, nodes are likely to use renewable energy sources, where devices and servers may be powered by solar~\cite{he2026malware}.  
This will likely become standard given the massive energy requirements of future networks, especially those that run AI models \cite{AIEnergy23}.
Given the said so called {\em green} IoT network, a key issue is scheduling and embedding applications with dependent tasks on nodes; e.g., such an application may perform video analytics where it first acquires videos from devices with a camera followed by tasks such as object detection and recognition~\cite{10973078}.
Figure~\ref{toyFig} shows an example IoT network, where we aim to run two such applications.  Their corresponding tasks, which are run as virtual network functions (VNFs), are deployed at gateways and at servers.  
VNF-Cs run on sensor devices. VNF-Ps run on gateways or servers to process collected data, after which results are sent to the sink.   

\begin{figure}[tphb]
    \centering
    \includegraphics[width=0.47\textwidth]{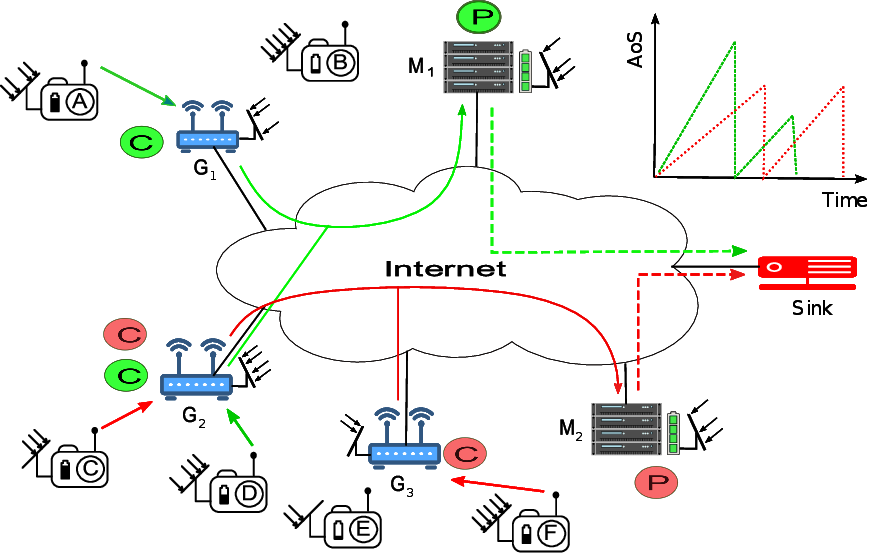}
    \caption{An example with server $M_1$ and $M_2$. The gateways ($G_1$ to $G_3$), devices, and servers rely on solar energy. The VNFs (denoted as a circle) of two DAGs have a different color, i.e., green and red.  These VNFs include (i) VNF-Cs, labeled as C in a circle, which are deployed at a gateway to download data from a device, and (ii) VNF-Ps, labelled as P in a circle, which process data from VNF-Cs.  Lastly, the traffic between VNFs are indicated by their corresponding arrow, and dotted arrows indicate result from a VNF-P.}
    \label{toyFig}
\end{figure}
When running applications, a key issue is ensuring they are run frequently so that they are able to acquire and process the latest data.  
Hence, we adopt age of service (AoS)~\cite{microserviceliu2022} as a performance metric, where the AoS of an application is equal to the time elapsed since the data generation time of the most recent service of the application to its completion time; i.e., an application with a low AoS equates to it having an up to date result.
Compared to the age of information (AoI)~\cite{FirstAoI}, AoS emphasizes the freshness of services, which include time taken to acquire and process data.
For example, a video analytics application may have one or more VNF-Cs to collect images from different cameras, which are then processed by a VNF-P that runs an object detection algorithm.   The VNF-P then sends its computed result to a sink node.  Here, the application's AoS corresponds to when a target is last (or not) detected in an area.

There are a number of challenges when scheduling and embedding applications in the aforementioned network.
{\em First}, solar energy arrivals exhibit spatio-temporal properties, meaning the energy level of nodes varies over time and space.
This has an impact on their ability to run a given VNF, or equivalently, the number of VNFs they can run in each time slot. 
Consequently, an energy outage may result in a VNF-C/VNF-P not being run in a time slot to collect or compute data, which leads to a higher AoS.  Moreover, a VNF may be run on a server that is far from the sink, which also increases its AoS.
The {\em second} issue is fair resource allocation, where a network must not starve an application in favor of other applications, meaning the AoS of an application must not grow significantly larger than the AoS of other applications. 
%
%

%

Henceforth, this paper makes the following contributions:
\begin{itemize}
    \item[C1] It is the first to study embedding of applications in a solar-powered network; as we will highlight in Section~\ref{RWORKS}, our problem is open where past works have {\em not} jointly considered AoS, energy harvesting servers, and applications with dependent tasks.  
    \item[C2] It outlines a novel mixed integer linear program (MILP) that can be used to optimally embed the tasks of applications over a planning horizon, where it aims to minimize the maximum AoS or min-max AoS of these applications.  
    The MILP is the first to jointly optimize (i) device selection and sampling time, (ii) embedding of dependent tasks, and (iii) energy usage of devices, gateways and servers.
    \item[C3] It outlines the first heuristic solution called GreedyOL, which embeds applications according to their AoS.  
    In addition, it presents a novel receding horizon control (RHC) approach called RHCOP that uses causal information to embed applications.     
    \item[C4] It reports on the first study of a novel system, problem and solutions.  
    It presents a number of theoretical properties, including the said problem's hardness, and computation complexity of MILP and GreedyOL.
    Then it presents numerical results that showed that the min-max AoS of RHCOP and GreedyOL is respectively 1.07x and 1.13x higher than MILP.  
    Additional gateways, servers and a larger solar panel size help reduce min-max AoS.  By contrast, more applications led to larger min-max AoS values.
    Further, when applications have equal number of VNF-Cs and VNF-Ps, they the lowest min-max AoS as compared to the case when they have more VNF-Cs versus VNF-Ps, and vice-versa.  
    %
    %
\end{itemize}
%

Next, we discuss prior works.  Section \ref{SMODEL} presents our network model, and Section \ref{PROBLEMF} defines the problem formally.  Section \ref{SOLS} outlines our RHC-based solution.  Section \ref{Heur} presents the heuristic solution.    
Section \ref{ANA} lists some properties of the formulated MILP and proposed solutions.
Our evaluation and results are discussed in Section \ref{EVAL}.  
Section \ref{CONC} concludes the paper.
\section{Related Works}\label{RWORKS} 
Our research overlaps with works that aim to (i) embed service function chains (SFCs) or applications modeled as directed acyclic graphs (DAGs) in a network,  (ii) minimize some freshness related metric, 
and (iii) age of network service.
\subsection{SFC/DAGs Embedding} 
Past works have considered the optimal mapping of virtual networks onto a network substrate.  For example, the work in \cite{VNE_Energy_2021_4aa} aims to minimize the energy consumption of wireless sensor networks (WSNs), whereas the work in \cite{xie2022joint} considers multi-access edge computing (MEC) 
to name a few.  However, the nodes in these works are not powered by a renewable energy source. 
%
By contrast, the authors of \cite{My_CL, Honglin_IoT} consider the problem of embedding virtual networks onto a radio-frequency (RF) powered Internet of things (IoT) network.
Both works aim to optimize the charging power used by a hybrid access point, VNFs placement, routing and/or link schedule.  
The work in \cite{cheng2020joint} considers sharing of nodes between different operators.  These nodes are also powered by RF and support virtualization.  The authors formulated a Stackelberg game that jointly controls node sharing, the charging strategy of an RF transmitter and caching in order to maximize profit.
Liang et al. \cite{liang2024sustainable} proposed a solution to minimize the energy cost of their network by optimizing VNF placement of SFCs and traffic routing.  The key challenges include time-varying wireless links and uncertain energy arrivals. 
However, the aforementioned works do not aim to minimize AoI or AoS, where they aim to maximize the number of embedded SFCs/VNFs.  They also consider a network with a dedicated power source as opposed to a renewable energy source.
%

%
\subsection{Information Freshness}\label{RWORKSB}
To date, many works have proposed metrics that aim to quantify the freshness of collected data. Examples include~\cite{he2022age}, and \cite{xiao2023aoi}.
We emphasize that our work is different to AoI works, which do not consider computation and DAGs.
To this end, we only consider works that aim to ensure the freshness of {\em computed} results.   In this respect, a number of works have considered devices that collect and process data that is required by users.  In this respect, timeliness of processed data is critical, where data is processed by a device or offloaded to a server.

Reference \cite{10155465} jointly considers the problem of executing sensed data locally at devices or offloaded to a cloud server in order to minimize AoI.   In the latter case, a key problem is scheduling data transmissions.
Ling et al. \cite{DynAgeAoIChen} consider devices that have energy capability.  These devices consider their available energy when processing data locally or offloading data to a server in order to minimize AoI.  
The work in \cite{9896777} considers devices that decide whether to process data locally or offloaded to a server to drive a co-located actuator.  
In \cite{shi2024enhancing}, the authors consider tasks associated with autonomous driving, where their aim is to ensure these tasks have the latest sensing data for vehicle control. They design a reinforcement learning-based solution to optimize sensor data selection and task placement, aiming to minimize the maximum AoI of tasks.  
Li et al. \cite{AoP222} proposed the age of processing (AoP) metric, which they then study in a system with a single server and server.  They consider sampling frequency as well as task offloading in order to minimize AoP.
In \cite{chen2021information}, ground users contend for a wireless channel and computing resources at an edge server and a unmanned aerial vehicle (UAV).    Their aim is to achieve a trade-off between information freshness and system energy cost. 
\subsection{Age of Network Service}\label{RWORKSC}
Some works have studied AoI and network virtualization.   
For example, reference~\cite{SFC_AoI} aims to minimize the energy cost of embedding SFCs whilst ensuring the AoI of each SFC remains below a threshold at all times.    
In \cite{JSACAoI}, each user has a SFC that involves data collection.  The authors of \cite{JSACAoI} aim to minimize the weighted sum of SFC embedding cost and the maximum AoI of users, where the AoI of a user is determined by the VNFs schedule of its SFC.
%
References \cite{peng2023aoi, 10449431} study partial offloading problem with task dependencies.  Its aim is to minimize the weighted sum of AoI, task processing latency and energy consumption.  Both references, however, only consider a linear directed graph; i.e., a sequence of tasks.
%

Reference \cite{microserviceliu2022} uses AoS to measure the freshness of service in industrial applications.  It embeds SFCs from different devices, and aims to minimize the average AoS of these SFC requests.  
%
%
In \cite{TimelyWakisaka2021}, the execution of serverless functions at servers relies on  the freshness of data from devices related to the function.  
It jointly optimizes data update at devices and data reception of servers to minimize the weighted sum of AoS of serverless functions.
%

\begin{table*}[htbp!]
    \centering
\caption{A comparison between works.}
\label{LR_Comp}
    \begin{threeparttable}
    %
    %
    \begin{tabular}{|c|c|c|c|c|c|c|c|} \hline 
         \textbf{Ref}&\textbf{\tabincell{c}{Only VNF \\ Placement}}&\textbf{\tabincell{c}{SFC \\ Embedding}}& \textbf{\tabincell{c}{DAG \\ Embedding}}&  \textbf{\tabincell{c}{Computation Resource\\Allocation}}&  \textbf{\tabincell{c}{Energy \\Harvesting}}&  \textbf{\tabincell{c}{Information \\Freshness}}& \textbf{\tabincell{c}{Data\\Collection}}\\ \hline 
         \cite{microserviceliu2022}&&\checkmark &  &\checkmark  &  &\checkmark  &\\ \hline
         \cite{VNE_Energy_2021_4aa}&\checkmark&&  &  &  &  & \\ \hline
         \cite{xie2022joint}&\checkmark&&  &\checkmark  &  &  & \\ \hline
         \cite{My_CL}& & &\checkmark  &\checkmark  &\checkmark  &  & \\ \hline
         \cite{Honglin_IoT}\cite{liang2024sustainable}& &\checkmark &  &  &\checkmark  &  &   \\ \hline
         \cite{cheng2020joint}&    & &\checkmark   &  &\checkmark   &  & \\ \hline
         %
         %
         %
         \cite{he2022age}&& &  &\checkmark  &\checkmark  &\checkmark  & \\ \hline
         \cite{xiao2023aoi}&   & &  &  &  &\checkmark  &\checkmark \\ \hline
         %
         %
         %
         %
         %
         \cite{DynAgeAoIChen}&& &  &\checkmark  &\checkmark  &\checkmark  & \\ \hline
         %
         %
         \cite{10155465}\cite{9896777}\cite{shi2024enhancing}& & &  &\checkmark  &  &\checkmark  &\checkmark \\ \hline
         \cite{AoP222}\cite{chen2021information}&& &  &\checkmark  &  &\checkmark  & \\ \hline
         %
         \cite{SFC_AoI}\cite{JSACAoI}&&\checkmark &  &\checkmark  &  &\checkmark  & \\ \hline
         \cite{peng2023aoi}&& &\checkmark  &\checkmark  &  &\checkmark  & \\ \hline
         \cite{TimelyWakisaka2021}&& &\checkmark  &\checkmark  &  &\checkmark  &\checkmark\\ \hline
         {\bf Our Work} &\checkmark   &\checkmark &\checkmark  &\checkmark  &\checkmark  &\checkmark  &\checkmark \\ \hline
    \end{tabular}
    
    \end{threeparttable}
\end{table*}

\subsection{Discussion}
Table \ref{LR_Comp} compares prior works, where they have considered embedding VNFs/SFCs in energy harvesting networks, namely \cite{My_CL}\cite{Honglin_IoT}\cite{cheng2020joint}\cite{liang2024sustainable}, do not have the same system setup nor address the same problem.  More specifically, these works assume there are dedicated power sources that can be controlled to deliver energy to nodes.  In contrast, the nodes in our network are ambient sources that exhibit spatio-temporal properties.  Further, these works do not consider scheduling of DAGs to minimize AoS.  In addition, they do not consider data collection from energy harvesting nodes.
The works in Section~\ref{RWORKSB} do not consider energy harvesting servers nor embedding of DAGs that include tasks that require one or more energy harvesting devices to collect data.
Lastly, the works in Section~\ref{RWORKSC} do not consider energy harvesting nodes.

In summary, although prior works have considered (computed) information freshness and embedding of SFCs/VNFs, they have not {\em jointly} considered the following aspects:
\begin{itemize}
\item End-devices and servers with energy harvesting capabilities, which govern the amount of data collected and number of applications run over time.   
\item Orchestrating applications with dependent tasks, where tasks may involve both data collection and computation.   This generalizes prior works that only consider a single sensor device or/and application.  
\item Unlike prior works, the AoS of applications is a function of the time varying energy and computation resources at nodes.  In particular, an application can only run when all its VNFs have sufficient resources to run, and the amount of resources consumed by an application are impacted by other applications.
\end{itemize}
%

%
%
\section{System Model}\label{SMODEL}
\begin{table}[htbp]
\small
\centering
\caption{\label{notationTabel}A summary of notations.}
\begin{tabular}{lp{7cm}}
\toprule
{\bf 1.} & {\bf Sets} \\
\midrule
%
$\mathcal{T}$ & Set of time slots. \\
$\mathcal{G}$ & Graph.  \\
$\mathcal{V}$ & Set of nodes.  \\
$s$			  & The sink.\\
$\mathcal{E}$ & Set of links.  \\
$\mathcal{V}_{D}$ & Set of devices.  \\
$\mathcal{V}_{G}$ & Set of gateways.  \\
$\mathcal{V}_{S}$ & Set of servers.  \\
$\mathcal{N}_{i}$ & Set of devices associated to gateway $i$.  \\
$\mathcal{E}_{W}$ & Set of wireless links.  \\
$\mathcal{E}_{G}$ & Set of wired links between gateways and servers.  \\
$\mathcal{E}_{S}$ & Set of wired links between servers and the sink.  \\
$\mathcal{R}$ & Set of DAG requests.  \\
$\mathcal{G}^{r}$ & DAG of request $r$.  \\
$\mathcal{V}_{G}^{r}$ & Set of VNF-Cs in DAG $r$.  \\
$\mathcal{V}_{S}^{r}$ & Set of VNF-Ps in DAG $r$.  \\
$\mathcal{E}^{r}$ & Set of direct links in DAG $r$.  \\
\bottomrule
{\bf 2.} & {\bf Constants} \\
\toprule
$\delta$ & Duration of a time slot. \\
$E_{n}^{max}$ & Battery capacity of node $n$. \\
$C_{n}^{max}$ & Processing capacity of node $n$. \\
$g_{di}^{t}$ & Channel power gain between device $d$ and gateway $i$. \\
$B$  & Bandwidth. \\
$\bar{R}$  & Capacity of a wired link. \\
$W_{n}^{t}$  & The amount of solar energy arrival. \\
$\Gamma_{ru}$  & Data rate requirement of VNF $u$ of DAG $r$. \\
$C_{u}^{r}$  & Processing requirement of VNF $u$ of DAG $r$.\\
$B_{uv}^{r}$  & Bandwidth requirement of link $(u,v)$ of DAG $r$.\\
$\Phi_{rui}$ & A binary variable to indicate the location of VNF-C.\\
$P_{i}^{0}$  & Base power of node $i$.\\
$P_{i}^{1}$  & Peak power of node $j$.\\
$\rho$  & Energy consumption of sensing one bit.\\
$\Psi$ & A large positive value to disable constraints.\\
\bottomrule
{\bf 3.} & {\bf Variables} \\
\toprule
$w_{n}^{t}$ & The amount of energy stored at node $n$ in slot $t$.  \\
$e_{n}^{t}$ & Energy level at node $n$ in slot $t$.  \\
$x_{rui}^{t}$ & A binary variable to indicate the activation of VNF-C. \\ 
$U_{i}^{t}$ & CPU utilization of gateway $i$ is active in slot $t$. \\
$\phi_{di}^{t}$ & A binary variable to indicate the activation of device.\\
$P_{di}^{t}$  & Transmit power of device $d$ at time slot $t$. \\
$y_{rvj}^{t}$ & A binary variable to indicate the activation of VNF-P.\\
$l_{ruv}^{tis}$ & A binary variable shows when edge $(u,v)$ of DAG $r$ uses the link $(i,s)$ in $t$.\\ 
$z_{ros}^{t}$ & A binary variable to indicate the activation of VNF-M.\\
$a_{r}^{t}$ & AoS value of DAG request $r$ in slot $t$.  \\
$\lambda_{r}^{t}$ & A non-negative variable.\\
\bottomrule
\end{tabular}
\end{table}

We optimize over discrete time slots in set $\mathcal{T}$; each time slot has index $t$ and duration $\delta$.  
For convenience, we will assume $\delta$ is one second, meaning the term energy and power are used interchangeably.
The network substrate is represented as a graph $\mathcal{G}(\mathcal{V}\cup \{o\},\mathcal{E})$, where $\mathcal{V}$ and $\mathcal{E}$ denote the set of nodes and links.  The sink has label $o$.  
Each node $n \in \mathcal{V}$ has battery capacity $E_{n}^{max}$.  
There are three types of substrate nodes: (i) devices, (ii) gateways, and (iii) servers.  We will use $d$ for device, $i$ for gateway and $s$ for server. 
Let $\mathcal{V}_{D}$, $\mathcal{V}_{G}$ and $\mathcal{V}_{S}$ denote the set of devices, gateways and servers, respectively.  
Each device $d$ is only associated to one gateway in $\mathcal{V}_{G}$.    Define $\mathcal{N}_{i}$ as the set of devices connected to gateway $g$. 
In addition, gateways assign an orthogonal frequency to their associated devices.
Note that frequency/channel assignment is beyond the scope of this paper.
The respective processing capacity of each gateway $i \in \mathcal{V}_{G}$ and server $s \in \mathcal{V}_{S}$ is $C_{i}^{max}$ and $C_{s}^{max}$. 
%

\subsection{Interconnections}
The wireless link between devices and their respective gateway is recorded in set $\mathcal{E}_{W}$.
We model these links as block fading Rician channels, where the gain of a channel is fixed during a slot but varies between two adjacent slots.  Let $g_{di}^{t}$ be the channel power gain between device $d$ and gateway $i$ in slot $t$, which is given by 
\begin{eqnarray}\label{Channel_Model}
    g_{di}^{t} = h C_0 \left(\frac{D_{di}}{D_0}\right)^{-\alpha}, &\forall (d,i)\in \mathcal{E}_{W},
\end{eqnarray}
where $h$ is drawn from an Exponential distribution with unity mean, $C_{0}$ indicates the path loss at a reference distance $D_{0}$, $\alpha$ is the path loss exponent, and $D_{di}$ is the Euclidean distance between device $d$ and gateway $i$.
There are two types of wired channels: (i) gateways-servers, and (ii) servers-sink $o$.
These channels are recorded in set $\mathcal{E}_{G}$ and $\mathcal{E}_{S}$, respectively.
Without loss of generality, we assume these channels have capacity $\bar{R}$ (bps).  
Note that in practice such links can be established by a software defined network (SDN) controller, where it installs a flow entry in switches along a path between two nodes~\cite{SDNTE}.
%

\subsection{Energy Harvesting}
Let $W_{n}^{t}$ be the amount of solar energy that arrives at node $n \in \mathcal{V}$ in time slot $t$.  
This amount is governed by the Markov model proposed in~\cite{Data_Driven_Solar}.
Briefly, the model has four states, each corresponding to a weather condition, namely ``Excellent'', ``Good'', ``Fair'' and ``Poor'', which are denoted as $\omega_E$, $\omega_G$, $\omega_F$ and $\omega_P$, respectively.  For each state, there is a Gaussian distribution with a given mean $\mu_{\omega}$ and variance $\sigma^2_{\omega}$, where $\omega\in\{\omega_E, \omega_G, \omega_F, \omega_P\}$. 
Further, there is a transition probability between states, denoted as
$\mathbb{P}_{ab}$, where $a,b\in\{\omega_E, \omega_G, \omega_F, \omega_P\}$.
A key advantage of the Markov model is that the mean and variance of the Gaussian distribution of each state, and the transition probability between states can be set as per solar irradiance measurements from a test-bed.
Each node $n$ can only store up to $E^{max}_n$ amount of energy.  Let $w_{n}^{t}$ denote the amount of energy stored by node $n$ in slot $t$.  
This amount is constrained by (i) the amount of solar energy arrivals, and (ii) the available battery capacity of node $n$.  Let $e_{n}^{t}$ record the energy level of node $n$ at the end of slot $t$.  Formally, the value of $w_{n}^{t}$ is constrained by
\begin{eqnarray}
	w_{n}^{t} \leq W_{n}^{t}, & \forall n \in \mathcal{V}, \forall t \in \mathcal{T},\label{EH_Bound_1}\\
	w_{n}^{t} \leq E_{n}^{max} -e_{n}^{t}, & \forall n \in \mathcal{V}, \forall t \in \mathcal{T}. \label{EH_Bound_2}
\end{eqnarray}
%

\subsection{Applications}
We model applications as DAGs.
Define $\mathcal{R}$ as the set of requests to embed DAGs.  Each DAG request $r$ has three types of of tasks that are run as the following VNFs: VNF-C, VNF-P and VNF-M.
VNF-C $u$ of DAG $r$ is located at a given gateway.  Its responsibilities include (i) collecting data from a device at a rate of $\Gamma_{ru}$ (bit/s), (ii) aggregating these raw data into a sample, and (iii) sending the sample to a VNF-P. 
Each VNF-P, which runs at a server, combines and processes samples from one or more VNF-Cs before sending the result to the sink.  
The sink runs VNF-M, meaning each DAG request $r$ is rooted at sink $o$ to merge all data.  

%
We model each DAG as a graph $\mathcal{G}^{r}(\mathcal{V}_{G}^{r}\cup\mathcal{V}_{S}^{r}\cup\{\text{VNF-M}\},\mathcal{E}^{r})$, where $\mathcal{V}_{G}^{r}$, $\mathcal{V}_{S}^{r}$, and $\mathcal{E}^{r}$ represent the set of VNF-C, VNF-P and directed edges, respectively.
We will use $u$, $v$ and $m$ to refer to VNF-C, VNF-P and VNF-M, respectively.  
%
%
Let $C_{u}^{r}$ denote the processing requirement, i.e., CPU cycles, of VNF $u$ of DAG $r$, where $u\in\mathcal{V}_{G}^{r}\cup\mathcal{V}_{S}^{r}$.  
Each directed edge $(u,v) \in \mathcal{E}^r$ requires bandwidth $B_{uv}^{r}$.

Next, we outline constraints relating to the placement or activation of VNFs at gateways, servers and devices.
%

\subsection{Gateways}
The VNF-Cs of DAG requests are pre-loaded at gateways.  
Define an indicator $\Phi_{rui} \in \{0,1\}$, where we have $\Phi_{rui}=1$ if VNF-C $u$ of DAG $r$ is located at gateway $i$; otherwise, $\Phi_{rui}=0$.
Let $x_{rui}^{t}\in\{0,1\}$ denote whether VNF-C $u$ of DAG $r$ located at gateway $i$ is active in slot $t$; we have $x_{rui}^{t}=1$ ($x_{rui}^{t}=0$) when VNF-C is active (inactive) in slot $t$.
For each gateway $i$, its processing capacity is constrained as
\begin{eqnarray}\label{CPU_Capa}
	\sum_{r\in\mathcal{R}} \sum_{u\in \mathcal{V}_{G}^{r}} \Phi_{rui} x_{rui}^{t} C_{u}^{r} \le  C_{i}^{max}, &\forall i \in \mathcal{V}_{G},\forall t \in \mathcal{T}. 
\end{eqnarray}
Each VNF-C runs on one gateway.  We thus have
\begin{eqnarray}\label{CPU_Capb}
	\sum_{i\in\mathcal{V}_{G}} \Phi_{rui}x_{rui}^{t} \le  1, &\forall u \in \mathcal{V}_{G}^{r}, \forall r\in\mathcal{R}, \forall t \in \mathcal{T}. 
\end{eqnarray}
Next, we consider the energy expenditure at gateway $i$, which is dependent on its utilization.  It has base power $P_{i}^{0}$ and peak power $P_{i}^{1}$.  
The utilization at gateway $i$ is
\begin{eqnarray}\label{GWUTIL}
    U_{i}^{t} =  \frac{\sum_{r\in\mathcal{R}} \sum_{u\in \mathcal{V}_{G}^{r}} \Phi_{rui}^{t} x_{rui}^{t} C_{u}^{r}}{C_{i}^{max}}, \forall i\in\mathcal{V}_{G}, \forall t\in\mathcal{T}.
\end{eqnarray}
The energy level of gateway $i$ is constrained as
\begin{eqnarray}\label{GWEVOLa}
	0\le e_{i}^{t} \le E_{i}^{max}, &\forall i\in\mathcal{V}_{G}, \forall t\in\mathcal{T}.
\end{eqnarray}
The energy at gateway $i$ evolves as
\begin{eqnarray}\label{GWEVOLb}
	e_{i}^{t} = e_{i}^{t-1} + w_{i}^{t-1} - (P_{i}^{1} - P_{i}^{0})U_{i}^{t}, &\forall i\in\mathcal{V}_{G}, \forall t\in\mathcal{T}.
\end{eqnarray}
\subsection{Devices}
Device $d$ must meet two conditions before uploading any data.  First, its associated gateway $i$ is active in slot $t$.  Second, gateway $i$ selects device $d$ in slot $t$.  
Let $\phi_{di}^{t}\in\{0,1\}$ denote whether device $d$ uploads data to its associated gateway $i$ in slot $t$, where in the case of upload, we have $\phi_{di}^{t}=1$.  Otherwise, we have $\phi_{di}^{t}=0$.  Formally, we have 
\begin{eqnarray}
	\sum_{d\in\mathcal{N}_{i}}\phi_{di}^{t} \ge \frac{\sum_{\forall r\in\mathcal{R}} \sum_{\forall u \in\mathcal{V}_{G}^{r}}x_{rui}^{t}}{\sum_{r\in\mathcal{R}}|\mathcal{V}_{G}^{r}|}, &\forall i\in\mathcal{V}_{G}, \forall t\in\mathcal{T},\label{Device_Up1}\\
	\sum_{d\in\mathcal{N}_{i}}\phi_{di}^{t} \le 1, &\forall i\in\mathcal{V}_{G}, \forall t\in\mathcal{T},\label{Device_Up2}
\end{eqnarray}
where Eq.~\eqref{Device_Up1} ensures $\phi_{di}^{t}$ is non-zero when a VNF-C on gateway $i$ is active in slot $t$. 
Further, Eq.~\eqref{Device_Up2} ensures gateway $i$ selects at most one device in each time slot.
A device must have sufficient energy to support a VNF-C at its associated gateway.   In particular, it must have sufficient energy to sense and transmit at a given rate $\Gamma_{ru}$ to each VNF-C $u$ of DAG $r$.
To this end, let $\rho$ denote the  energy consumption of sensing one bit.  
Define $P_{di}^{t}$ to be the transmission power required to achieve a rate of $\Gamma_{ru}$, which is calculated as
\begin{eqnarray}
	P_{di}^{t} = \frac{(2^{\Gamma_{ru}/B} - 1)N_0}{g_{di}^{t}}. \label{DEVTXPOWa}
\end{eqnarray}

The total energy cost of device $d$, which includes transmission and sensing, in time slot $t$ is $\phi_{di}^{t}(P_{di}^{t} + \rho\Gamma)$.  
For each $t\in\mathcal{T}$, the energy at each device $d \in \mathcal{V}_{D}$ associated to each gateway $i\in\mathcal{V}_G$ evolves as
\begin{eqnarray}\label{DEVOL}
	e_{d}^{t} = e_{d}^{t-1} + w_{d}^{t-1} - \phi_{di}^{t}(P_{di}^{t} + \rho\Gamma_{ru}). \label{DEVTXPOWb}
\end{eqnarray}
%

\subsection{Servers}
The servers in $\mathcal{V}_{S}$ run VNF-Ps.  
Let the binary variable $y_{rvs}^{t}\in\{0,1\}$ denote whether VNF-P $v$ of DAG $r$ is mapped to server $s$ in slot $t$.  
As the computation resource of server $s$ is finite, we have
\begin{eqnarray}\label{ServerCPU}
	\sum_{r\in\mathcal{R}} \sum_{v\in \mathcal{V}_{S}^{r}} y_{rvs}^{t} C_{v}^{r} \le  C_{s}^{max}, &\forall s \in \mathcal{V}_{S},\forall t \in T. 
\end{eqnarray}

Each VNF-P of a DAG runs on one server:
\begin{eqnarray}\label{OneServer}
	\sum_{s\in\mathcal{V}_{S}} y_{rvs}^{t} \le 1, &\forall v\in \mathcal{V}_{S}^{r}, \forall r\in\mathcal{R}, \forall t\in\mathcal{T}.
\end{eqnarray}
%

Next, we consider the energy consumed by server $s$.
Let its base power be $P_{s}^{0}$, and peak power when fully utilized is $P_{s}^{1}$.
Define the processing utilization of server $s$ as
\begin{eqnarray}\label{SUTIL}
	U_{s}^{t} =  \frac{\sum_{r\in\mathcal{R}} \sum_{V\in \mathcal{V}_{S}^{r}} y_{rvs}^{t} C_{v}^{r}}{C_{s}^{max}}, &\forall s\in\mathcal{V}_{S}, \forall t\in\mathcal{T}.
\end{eqnarray}
The energy consumed $e_{s}^{t}$ by server $s$ is by
\begin{eqnarray}\label{GWEVOL3}
	0\le e_{s}^{t} \le E_{s}^{max}, &\forall s\in\mathcal{V}_{S}, \forall t\in\mathcal{T}.
\end{eqnarray}
The energy evolution of server $s$ during each slot $t\in\mathcal{T}$ is 
\begin{eqnarray}\label{GWEVOL4}
	e_{s}^{t} = e_{s}^{t-1} + w_{s}^{t-1} - (P_{s}^{1} - P_{s}^{0})U_{s}^{t} \ge 0, &\forall s\in\mathcal{V}_{S}.
\end{eqnarray}
\subsection{Sink}
We assume sink $o$ has sufficient energy and computing resource.   Sink $o$ runs VNF-M of DAG $r$ only when all VNFs in DAG $r$ are active. 
Let $z_{rmo}^{t} \in \{0,1\}$ denote whether sink $o$ runs VNF-M of DAG $r$ ($z_{rmo}^{t}=1$) or not ($z_{rmo}^{t}=0$). 
To this end, we have 
\begin{eqnarray}
	z_{rmo}^{t} \le \sum_{i\in \mathcal{V}_{G}}x_{rui}^{t}, &\forall r \in \mathcal{R}, \forall u\in \mathcal{V}_{G}^{r}, \forall t \in \mathcal{T},\label{Sink_1}\\
	z_{rmo}^{t} \le \sum_{s\in \mathcal{V}_{S}}y_{rvs}^{t}, &\forall r \in \mathcal{R}, \forall v\in \mathcal{V}_{S}^{r}, \forall t \in \mathcal{T}.\label{Sink_2}
\end{eqnarray}

\subsection{Routing}
The total traffic routed between a gateway and a server, and between a server and sink $o$ must not exceed $\bar{R}$.
Define $l_{ruv}^{tis}\in\{0,1\}$; it is set to $l_{ruv}^{tis}=1$ when edge $(u,v)$ of DAG $r$ uses the link between gateway $i$ and server $s$ in slot $t$.  Otherwise, it is zero.
Further, we have $l_{ruv}^{tis}=1$ when (i) VNF-C of DAG $r$ is active at gateway $i$ in slot $t$, i.e., $x_{rui}^{t}=1$, (ii) VNF-P of DAG $r$ is mapped to server $s$ in slot $t$, i.e., $y_{rvs}^{t}=1$.  
To this end, for each DAG request $r \in \mathcal{R}$ in each slot $t\in\mathcal{T}$, we have
\begin{eqnarray}\label{Routing_1a}
	x_{rui}^{t} = \sum_{s\in\mathcal{V}_{S}} l_{ruv}^{tis}, &\forall (u,v)\in \mathcal{E}^{r}, \forall i\in\mathcal{V}_{G}, \label{Link_1}\\
	\sum_{i\in\mathcal{V}_{G}} l_{ruv}^{tis} = y_{rvs}^{t},  &\forall (u,v)\in \mathcal{E}^{r}, \forall s\in\mathcal{V}_{S}. \label{Link_2}
\end{eqnarray}
In addition, VNF-C $u$ and VNF-P $v$ of DAG $r$ run only when the path between gateway $i$ and server $s$ has sufficient bandwidth. Formally, we have 
\begin{eqnarray}
	\sum_{r\in\mathcal{R}}\sum_{u\in \mathcal{V}_{G}^{r}}\sum_{v\in \mathcal{V}_{S}^{r}}  l_{ruv}^{tis}B_{uv}^{t} \leq \bar{R}, &\forall (i,s) \in \mathcal{E}_{G}, \forall t\in\mathcal{T}. \label{Link_3}
\end{eqnarray}
Lastly, recall each server has only one link to sink $s$, the value of $y_{rvs}^{t}$ indicates whether edge $(v,m)$ is active on link $(s,o)$ in slot $t$.   Thus, for each slot $t\in\mathcal{T}$, we have 
\begin{eqnarray}\label{Routing_1b}
	\sum_{r\in\mathcal{R}}\sum_{v \in \mathcal{V}_{S}^{r}} y_{rvo}^{t}B_{vm}^{t} \leq \bar{R}, &\forall s \in \mathcal{V}_{S}, \forall t\in\mathcal{T}, \label{Link_4}
\end{eqnarray}
where edge $(v,o)$ is active when the channel between server $s$ and sink $o$ has sufficient bandwidth resource. 
%

\subsection{Age of Service (AoS)}
We assume the {\em generate-at-will} and {\em just-in-time} policy~\cite{Elmagid_AoI}, where each device generates a sample for a given DAG if all its VNFs are allocated resources to run in a given time slot. This means the VNF-Ps of a DAG will receive all required samples from VNF-Cs and generate a result to sink $o$.   Note that for each DAG $r$, a sample is generated by each of its VNF-C only after sink $o$ receives the result from its VNF-P(s).
We emphasize that AoS under {\em generate-at-will} is different with transmission delay.  This is because AoS accounts for the entire time since the generation time of the most recent processing update, whereas transmission delay only refers to the transmission time of data from a source to a destination.

As an aside, we note that metrics such as delay and AoI are not suitable for our problem.  
%
%
The reason is because it neglects (i) computation process, and (ii) multiple data sources.
To incorporate both factors, we define the AoS of an application as being equal to the time elapsed since the data generation time of the most recent service of the application to its completion time.  It accounts for the entire time elapsed since the data generation of the most recent service completion. It specifically emphasizes the freshness of a service completion, which includes data collection, multi-stage VNF computation, and delivery of the final result.  
We denote $a_{r}^{t}$ as the AoS value of DAG request $r$.  
It is set to one when DAG $r$ is active in slot $t$, i.e., $z_{rmo}^{t}=1$.  Otherwise, the AoS of DAG $r$ increases by one.  
The value of AoS evolves as per
\begin{eqnarray}
	a_{r}^{t} = (1-z_{rmo}^{t})a_{r}^{t-1}+1, &\forall r \in \mathcal{R}, \forall t \in \mathcal{T}. \label{AoI}
\end{eqnarray}
Thus, we have $a_{r}^{t}=1$ when VNF-M of DAG $r$ is active in slot $t$, meaning sink $s$ received a sample from DAG $r$ in slot $t$.  Otherwise, $a_{r}^{t}$ increases by one.  
%

\section{Problem Formulation}\label{PROBLEMF}
We aim to optimize the min-max AoS of $|\mathcal{R}|$ DAG requests over $|\mathcal{T}|$ slots.
Before presenting our MILP, we have to first linearize Eq.~\eqref{AoI}, i.e., $z_{rmo}^{t}a_{r}^{t-1}$.   
To do so, we replace $z_{rmo}^{t}a_{r}^{t-1}$ in Eq.~\eqref{AoI} with a new variable $\lambda_{r}^{t}$.  Let $\Psi$\footnote{This is also called the big-M method.} be a large positive integer value.  
We note that the AoS of request $r$ increases by one in each time slot if it is not served.  Therefore, the maximum value of $a_{r}^{t}$ is $|\mathcal{T}|$ over a finite planning horizon of $|\mathcal{T}|$ slots.  Hence, it is appropriate to set $\Psi = |\mathcal{T}|$.
We also introduce the following constraints to set the value of $\lambda_{r}^{t}$:
\begin{eqnarray}  
	0 \leq  \lambda_{r}^{t} \leq  z_{rmo}^{t} \Psi, &\forall r \in \mathcal{R}, \forall t \in \mathcal{T},\label{Node_1}\\  
	\lambda_{r}^{t}\geq a_{r}^{t-1} - (1 - z_{rmo}^{t}) \Psi,  &\forall r \in \mathcal{R}, \forall t \in \mathcal{T}, \label{Node_2}  \\
	\lambda_{r}^{t}\leq a_{r}^{t-1} +(1 - z_{rmo}^{t}) \Psi,  &\forall r \in \mathcal{R}, \forall t \in \mathcal{T}, \label{Node_3}\\
	a_{r}^{t} = a_{r}^{t-1} - \lambda_{r}^{t} + 1, &\forall r \in \mathcal{R}, \forall t \in \mathcal{T}. \label{Node_4} 
\end{eqnarray}
Observe that Eq.~\eqref{Node_1}-\eqref{Node_4} present two scenarios.
The first scenario is when DAG $r$ is inactive ($z_{rmo}^{t}=0$), which forces $\lambda_{r}^{t}$ to zero, see Eq.~\eqref{Node_1}.  Further, Eq.~\eqref{Node_2} and Eq.~\eqref{Node_3} are disabled.  In addition, the value of $a_{r}^{t}$ increases by one when we have $z_{rmo}^{t}=0$, see Eq.~\eqref{Node_4}.
The second scenario is when DAG $r$ is active ($z_{rmo}^{t}=1$).  We use Eq.~\eqref{Node_2} and Eq.~\eqref{Node_3} to force $\lambda_{r}^{t}$ to $a_{r}^{t-1}$.  Consequently, as per Eq.~\eqref{Node_4}, the value of
$a_{s}^{t}$ becomes one.
Finally, we have our MILP:
\begin{align}
\label{milp}
    \min_{\phi_{di}^{t},x_{rui}^{t},y_{rvj}^{t},l_{ruv}^{tis},z_{rmo}^{t}} & \quad \text{max} \left\{\sum_{t\in\mathcal{T}}\frac{a_{r}^{t}}{|\mathcal{T}|}\right\}_{r\in\mathcal{R}}, \tag{$\mathbf{P1}$} \\
    \text{s.t.} \quad & \eqref{EH_Bound_1}\mbox{-}\eqref{CPU_Capb}, \eqref{GWEVOLa}\mbox{-}\eqref{Device_Up2}, \eqref{DEVOL}\mbox{-}\eqref{OneServer},\nonumber\\ &\eqref{GWEVOL3}\mbox{-}\eqref{Link_4},\eqref{Node_1}\mbox{-}\eqref{Node_4}.\nonumber
\end{align}

We conclude with the following remarks.
First, $\delta$ represents one unit of time, which is defined by an operator; note, our formulation remains the same for any unit of time.  However, care has to be taken in terms of sizing resources such as energy and computation to reflect a chosen new time duration.
Second, the proposed MILP can be solved by a commercial solver for small problem instances.   As shown in Proposition \ref{Pro1}, the proposed MILP becomes intractable in large scale networks.  
Further, VNF Placement and Traffic Routing (VPTR) problem \cite{NFV21}, an NP-hard problem, is a special case of our problem, see Proposition~\ref{Pro0}.  
%
Lastly, non-causal information, i.e., future wireless channel coefficients and solar energy arrival, is required to solve \ref{milp}, meaning it is not practical and serves only as a theoretical benchmark.  
In the following sections, we present solutions that can be run by a central controller, where the first solution uses historical information, and the second solution only uses current resource information at devices, gateways and servers.  
Next, we present a practical solution that only requires historical information.
%

\section{A Receding Horizon Solution}\label{SOLS}
We design an RHC based solution called Receding Horizon Solution Optimization (RHCOP).
Briefly, RHC is a well-known feedback control strategy.  It builds a time horizon/window, and employs a prediction model to estimate future information within the said window.  
It then optimizes over the time slots in the window and adopts the computed solution for the current time slot.  It then shifts the said window by one time slot, and repeats the process.  

In our case, RHCOP constructs a time window with $K$ time slots.  
It employs a Gaussian mixture model (GMM) to estimate the probability distribution of wireless channel gains and solar energy arrivals at each node over $K$ time slots.
The reason of selecting GMM is that the solar energy arrivals at devices can be characterized as a Markov model with four states.   Each state has a Gaussian distribution with its own mean and variance.  Consequently, energy arrivals are governed by a mixture of Gaussian distribution, i.e., a GMM.

Using the said GMM, RHCOP determines the energy arrivals over the said $K$ time slots, and formulates MILP \eqref{milp} that optimizes over these $K$ time slots.  It then solves the MILP, and applies the computed decision for slot $t$, i.e., $\phi_{di}^{t}$, $x_{rui}^{t}$, $y_{rvj}^{t}$, and $z_{ros}^{t}$.  
Lastly, it shifts the window by one time slot to cover slot $t+1, \ldots, K+1$, and repeats the aforementioned steps to obtain the decision for time $t+1$. 
Figure \ref{RHCExample} gives an overview of RHCOP.
\begin{figure}[htbp]
	\centering 
	\includegraphics[scale = 0.5]{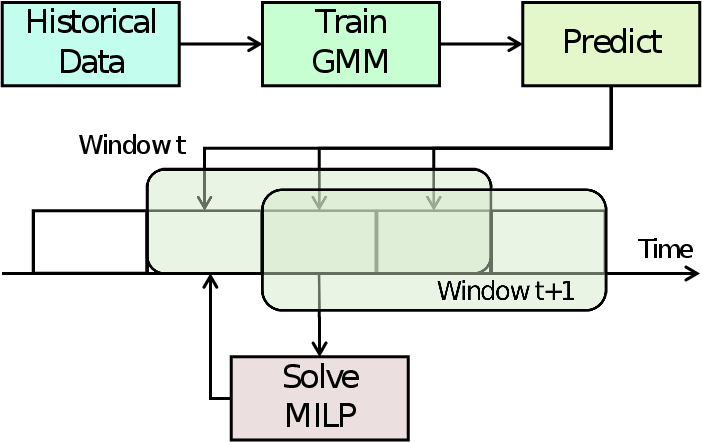}
	\caption{An overview of RHCOP in a network with a single device and server. }
	\label{RHCExample}
\end{figure}
RHCOP has three phases: {\em Training}, {\em Prediction}, and {\em Decision}.  
In the {\em Training} phase, it collects historical data of (i) channel gains for links in $\mathcal{E}_{W}$, (ii) solar arrivals of nodes in $\mathcal{V}$.  
For both (i) and (ii), RHCOP trains the GMM of each node using historical data; note, there is a GMM for each wireless link in $\mathcal{E}_{W}$ and node in $|\mathcal{V}|$.  
Briefly, each GMM includes a finite number of weighted Gaussian distributions/components with a corresponding mean and variance,  where the mean and variance of each distribution is computed by the Expectation-Maximization (EM) algorithm \cite{EM}.
We use $N_{G}$ to record the number of Gaussian components of GMM. 
Note, instead of GMM, other methods, such as a recurrent neural networks, can be used without changing our solution.
The {\em Training} phase can also be scheduled periodically to incorporate new data that improves the GMM of nodes.
%

%
Each slot $t$ contains two phases: (i) {\em Prediction}, and (ii) {\em Decision}.
In the {\em Prediction} phase, RHCOP constructs a time horizon window with size $K$.  
It acquires channel gain at the beginning of slot $t$ via standard pilot-based estimation techniques.
It then calls the GMM of each device, gateway or server to obtain an energy arrival estimate and channel gain estimate of time slots in the said window.    
In the {\em Decision} phase, RHCOP solves MILP \ref{milp} over the given time window, and applies the solution in slot $t$.  
After that, in time slot $t+1$, RHCOP shifts its time window by one slot, and repeats the previous steps.
Algorithm \ref{RHCOP} shows the details of RHCOP. 
%
Each server in $\mathcal{V}_{S}$ and gateways in $\mathcal{V}_{G}$ collect a data set of their historical energy arrivals.
Let the corresponding data be stored in set $\mathbf{E}_{i}$ and $\mathbf{G}_{d}$, see lines 2-7. 
Let $Gg_{d}$ be the GMM of the wireless channel of device $d\in\mathcal{V}_{D}$.  Denote $GW_{n}$ as the GMM that models the energy arrivals at each node $n\in\mathcal{V}$.  RHCOP calls {\em Train(.)} to train GMM $Gg_{d}$ and $GW_{n}$, see lines 8-13.
After that, RHCOP runs the {\em Prediction} phase and the {\em Decision} phase in each slot, see lines 14-26.  
To do so, RHCOP builds a time window, denoted as $\mathcal{K}$, that has $K$ slots, see line 16.
Let $\{\bar{W}_{n}^{t}\}_{t\in \mathcal{K}}$ record the estimated energy arrivals at node $n$ over the said time window; similarly, denote $\{\bar{g}_{di}^{t}\}_{t\in \mathcal{K}}$ as the estimated channel gains of wireless link $(d,i)$.
In lines 17-22, RHCOP calls {\em PredictEH(.)} and {\em PredictCha(.)} to estimate the value of $\{\bar{W}_{n}^{t}\}_{t\in \mathcal{K}}$ and $\{\bar{g}_{di}^{t}\}_{t\in \mathcal{K}}$, respectively.
RHCOP runs the {\em Decision} phase, which solves MILP \ref{milp} over the fixed time window using {\em SolveMILP(.)}; it uses $\Phi$ to record the results of MILP \ref{milp}, see line 24.  
Lastly, RHCOP runs {\em UpdateParameters(.)} to update the value of $\phi_{di}^{t}$, $x_{rui}^{t}$, $y_{rvj}^{t}$, $z_{ros}^{t}$, and $e_{i}^{t}$, see line 25. 
\begin{algorithm}
	\KwInput{$\{\hat{g}_{di}^{t}\}_{\forall (d,i)\in \mathcal{E}_{W}}, \{\hat{W}_{n}^{t}\}_{\forall n \in \mathcal{V}}$} 
	\KwOutput{$\{a_{r}^{t}\}_{r\in\mathcal{R}}$}
    \For{$n\in\mathcal{V}$}
    {
        $\mathbf{E}_{n}$ = CollectData(n)
    }
    \For{$(d,i)\in\mathcal{E}_{W}$}
    {
        $\mathbf{G}_{di}$ = CollectData(d)
    }
    \For{$n\in\mathcal{V}$}
    {	
        $GW_{n}$ = Train($\mathbf{G}_{di}$)\\
    }
    \For{$(d,i)\in\mathcal{E}_{W}$}
    {
        $Gg_{d}$ = Train($\mathbf{E}_{n}$)\\
    }
    \For{$t \in \mathcal{T}$}
	{
        %
		$\mathcal{K}$ = $\{t,t+1\dots,t+K\}$\\
        \For{$n\in\mathcal{V}$}
        {
            $\{\bar{W}_{n}^{t}\}_{t\in \mathcal{K}}$ = PredictEH(n, $\mathcal{K}$, $GW_{n}$)\\
        }
        \For{$(d,i)\in\mathcal{E}_{W}$}
        {
            $\{\bar{g}_{di}^{t}\}_{t\in \mathcal{K}}$ = PredictCha(d, $\mathcal{K}$, $Gg_{d}$)
        }
		%
		$\Phi$ = SolveMILP($\mathcal{K}$, $\{\bar{W}_{n}^{t}\}_{t\in \mathcal{K}}$, $\{\bar{g}_{di}^{t}\}_{t\in \mathcal{K}}$)\\
        UpdateParameters($\Phi$, $t$)\\
	}
	\caption{Pseudocode of RHCOP.}
	\label{RHCOP}
\end{algorithm}
\section{A Heuristic Solution: GreedyOL}\label{Heur}
Greedy Online (GreedyOL) embeds and schedules DAG requests according to their AoS.  It has three key phases in each time slot: (i) {\em Sort}, (ii) {\em Embed}, and (iii) {\em Update}.
The {\em Sort} phase sorts all DAG requests in decreasing order of their AoS value.  The {\em Embed} phase aims to find gateways, servers and physical links that have sufficient computation, energy and communication resources to support DAGs.  The {\em Update} phase determines whether to activate a DAG using the result in the {\em Embed} phase.  
Next, we explain each phase in detail using Algorithm \ref{Greedy}, which is run in each time slot.

In the {\em Sort} phase, GreedyOL calls {\em SortDAoS(.)} to place DAG requests in set $\mathcal{R}$ in descending order of their AoS, i.e., $a_{r}^{t}$, where GreedyOL records the results in set $\bar{\mathcal{R}}$, see line 1.  
The {\em Embed} phase checks DAGs in $\bar{\mathcal{R}}$ using three steps.
First, it finds a gateway $i$ in set $\mathcal{V}_{G}$ for each VNF-C $u$ of DAG $r$, where gateway $i$ must satisfy three conditions: (i) gateway $i$ has sufficient energy to support VNF-C $u$, (ii) gateway $i$ has computation resource to run VNF-C $u$, and (iii) there is at least one device that has sufficient energy of collect data for gateway $i$.  
These conditions are checked by {\em GatewayECD(.)}, and it stores the suitable gateway in $\eta_{u}^{r}$.   After finding a gateway for VNF-C $u$ of DAG $r$, it moves to the next VNF-C of DAG $r$, see line 5-6.   
The second step of the {\em Embed} phase is to find a server $s$ for each VNF-P $v$ in each DAG request $r$.  Server $s$ must satisfy two conditions: (i) it has sufficient energy to embed VNF-P $v$, and (ii) it has sufficient CPU cycles for VNF-P $v$.  
GreedyOL calls ServerEC(.) to check the said conditions, and the corresponding server $s$ is stored in $\psi_{v}^{r}$.  After that, it moves to the next VNF-P, see line 12.   
The last step of the {\em Embed} phase relates to communication resource.   Denote $\gamma_{uv}^{r}$ as a binary variable associated with each edge of DAG $r$.  GreedyOL calls {\em CheckBand(.)} to check whether the link $\eta_{u}^{r}$ and $\psi_{v}^{r}$ have sufficient bandwidth to support the demand of edge (u,v) ($\gamma_{uv}^{r}=1$) or not ($\gamma_{uv}^{r}=0$).  

In the {\em Update} phase, GreedyOL uses $\gamma_{uv}^{r}$ and CheckDAG(.) to check whether all VNF-Cs and VNF-Ps of DAG $r$ are embedded successfully.  If yes, GreedyOL calls UpdateRe(.) to update the energy, computation and communication resource of gateways and servers, see line 20.  After that, GreedyOL embeds DAG $r$ and set its AoS to one.  Otherwise, the value of $a_{r}^{t}$ increases by one. 
\begin{algorithm}
	\KwInput{$\{g_{di}^{t}\}_{\forall (d,i)\in \mathcal{E}_{W}}, \{W_{n}^{t}\}_{\forall n \in \mathcal{V}}$} 
	\KwOutput{$\{a_{r}^{t}\}_{r\in\mathcal{R}}$}
    %
        $\bar{\mathcal{R}}$ = SortDAoS($\{a_{r}^{t}\}_{\forall r \in \mathcal{R}}$)\\
        \For{$r\in \bar{\mathcal{R}}$}
        {
            \For{$u \in \mathcal{V}_{G}^{r}$}
            {
                \For{$i \in \mathcal{V}_{G}$}
                {
                    $\eta_{u}^{r}$ = GatewayECD(i)\\
                    Break
                }
            }
            \For{$v\in\mathcal{V}_{S}^{r}$}
            {
                \For{$s \in \mathcal{V}_{S}$}
                {
                     $\psi_{v}^{r}$=ServerEC(s)\\
                     Break
                }
            }
            \For{$(u,v)\in\mathcal{E}^{r}$}
            {
                $\gamma_{uv}^{r}$ = CheckBand(u,v,$\eta_{u}^{r}$,$\psi_{v}^{r}$)\\
                Break
            }
            \eIf{CheckDAG(\{$\gamma_{uv}^{r}\}_{(u,v)\in\mathcal{E}^{r}}$)}
            {
                UpdateRe($\{\eta_{u}^{r}\}_{u\in\mathcal{V}^{G}},\{\psi_{v}^{r}\}_{v\in\mathcal{V}^{S}}, \{\gamma_{uv}^{r}\}_{(u,v)\in\mathcal{E}^{r}}$)\\
                $a_{r}^{t}$ = 1\\
            }
            {
                $a_{r}^{t}$ = $a_{r}^{t}$ + 1\\
            }
        }
    \caption{Pseudocode of GreedyOL.}
    \label{Greedy}
\end{algorithm}
\section{Analysis}\label{ANA}
Here, we present facts relating to the hardness of our problem, and computational complexity of MILP and GreedyOL.
%

\begin{prop}
    \label{Pro0}
    Problem \ref{milp} is at least NP-hard.  
\end{prop}

\begin{proof}
    We show that the NP-hard VPTR problem~\cite{NFV21} is a special case of our problem.  Consider a substrate network and a set of requests.  Specifically, each substrate node has limited energy and computational resources.  Each link has limited bandwidth capacity.   Each request includes a set of VNFs with specific energy, computational and bandwidth requirements.  
    VPTR aims to place VNFs on substrate nodes, and finds routes between two VNFs without violating node and link resource capacity.  
    To this end, we remove three types of constraints from our problem: (i) energy constraints at devices, gateways and servers, which include Eq.~\eqref{EH_Bound_1}, \eqref{EH_Bound_2}, \eqref{GWEVOLa}, \eqref{GWEVOLb}, \eqref{DEVTXPOWb}, \eqref{GWEVOL3}, \eqref{GWEVOL4}, (ii) device selection constraints at each gateway, which include Eq.~\eqref{Device_Up1} and \eqref{Device_Up2}, (iii) AoS constraints of each DAG request, meaning Eq.\eqref{Node_1}-\eqref{Node_4}.
    Thus, VRTP is a special case of our problem, meaning our problem is at least NP-hard. 
\end{proof}

\begin{prop}
    \label{Pro1}
    The MILP \ref{milp} has $|\mathcal{T}|(\sum_{r\in\mathcal{R}}(|\mathcal{V}_{G}^{r}|+|\mathcal{V}_{S}^{r}|+|\mathcal{E}^{r}||\mathcal{V}_{G}|+|\mathcal{E}^{r}||\mathcal{V}_{S}|)+|\mathcal{E}_{G}|+4|\mathcal{V}_{S}|+4|\mathcal{R}|+2|\mathcal{V}|+5|\mathcal{V_{G}}|+|\mathcal{V}_{D}|)$  constraints and $(|\mathcal{V}_{D}|+\sum_{r\in\mathcal{R}}|\mathcal{V}_{G}^{r}||\mathcal{V}_{G}|+\sum_{r\in\mathcal{R}}|\mathcal{V}_{S}^{r}||\mathcal{V}_{S}|+\sum_{r\in\mathcal{R}}|\mathcal{V}_{G}^{r}||\mathcal{V}_{S}^{r}||\mathcal{V}_{G}||\mathcal{V}_{S}|+|\mathcal{R}|)|\mathcal{T}|$ decision variables. 
\end{prop}
\begin{proof}
    We start with decision variables.  There are four types of decision variables, i.e., $\phi_{di}^{t}$, $x_{rui}^{t}$, $y_{rvs}^{t}$, $l_{ruv}^{tis}$ and $z_{rmo}^{t}$.  
    As each device is associated with at most one gateway in each time slot, we have $|\mathcal{V}_{D}||\mathcal{T}|$ of $\phi_{di}^{t}$.  
    Next, the variable $x_{rui}^{t}$ exists for each VNF-C of each DAG request at each gateway in each slot.  Thus, there are $\sum_{r\in\mathcal{R}}|\mathcal{V}_{G}^{r}||\mathcal{V}_{G}||\mathcal{T}|$ such variables.
    As for $y_{rvs}^{t}$, each VNF-P of each DAG request at each server in each slot has one $y_{rvs}^{t}$.  We thus have $\sum_{r\in\mathcal{R}}|\mathcal{V}_{S}^{r}||\mathcal{V}_{S}||\mathcal{T}|$ such variables.   
    The variable $l_{ruv}^{tis}$ exists for edge of each DAG request at each link in each slot.  Thus, we have $\sum_{r\in\mathcal{R}}|\mathcal{V}_{G}^{r}||\mathcal{V}_{S}^{r}||\mathcal{V}_{G}||\mathcal{V}_{S}||\mathcal{T}|$ such constraints.  
    Lastly, as the VNF-M of each DAG runs only on the sink in each slot, there are $|\mathcal{R}||\mathcal{T}|$ such variables.  
    In summary, We thus have $(|\mathcal{V}_{D}|+\sum_{r\in\mathcal{R}}|\mathcal{V}_{G}^{r}||\mathcal{V}_{G}|+\sum_{r\in\mathcal{R}}|\mathcal{V}_{S}^{r}||\mathcal{V}_{S}|+\sum_{r\in\mathcal{R}}|\mathcal{V}_{G}^{r}||\mathcal{V}_{S}^{r}||\mathcal{V}_{G}||\mathcal{V}_{S}|+|\mathcal{R}|)|\mathcal{T}|$ decision variables, as desired.

    We now quantify the number of constraints in MILP \ref{milp}.
    Each device, gateway and server has one constraint of type \eqref{EH_Bound_1} and \eqref{EH_Bound_2} in each slot, where there are $2|\mathcal{V}||\mathcal{T}|$ such constraints.
    Each gateway has one constraint of type \eqref{CPU_Capa}, \eqref{GWEVOLa} and \eqref{GWEVOLb}, \eqref{Device_Up1} and \eqref{Device_Up2}, resulting in $5|\mathcal{V}_{G}||\mathcal{T}|$ constraints.  
    Each device has one constraint of type \eqref{DEVTXPOWb} in each slot, which results in $|\mathcal{V}_{D}|$ constraints.
    Each server has one constraint of type \eqref{ServerCPU}, \eqref{GWEVOL3}, and \eqref{GWEVOL4} in each slot, we have $3|\mathcal{V}_{S}||\mathcal{T}|$ such constraints.
    The VNF-C of each DAG request has one constraint of type \eqref{CPU_Capb} and \eqref{Sink_1} in each slot, resulting in $\sum_{r\in\mathcal{R}}|\mathcal{V}_{G}^{r}||\mathcal{T}|$ such constraints.  
    The VNF-P of each DAG request has one constraint of type \eqref{OneServer} and \eqref{Sink_2}, resulting in $\sum_{r\in\mathcal{R}}|\mathcal{V}_{S}^{r}||\mathcal{T}|$ such constraints.
    Next, as for constraint \eqref{Link_1} and \eqref{Link_2}, they exist on one edge of each DAG in each slot for each gateway and server, respectively.  Thus, we have $\sum_{r\in\mathcal{R}}$ $|\mathcal{E}^{r}||\mathcal{T}|(|\mathcal{V}_{G}|+|\mathcal{V}_{S}|)$ such constraints.
    Constraint \eqref{Link_3} exists for each link between gateways and servers in each slot.  Thus, we have $|\mathcal{E}_{G}||\mathcal{T}|$ such constraints.
    Constraint \eqref{Link_4} exists for each link between servers and sink in each slot, where we have $|\mathcal{V}_{S}||\mathcal{T}|$ such constraints.
    Lastly, each DAG request has one constraint of type \eqref{Node_1}-\eqref{Node_4}.  Hence, we have $4|\mathcal{R}||\mathcal{T}|$,
    In summary, we have 
    $|\mathcal{T}|(\sum_{r\in\mathcal{R}}(|\mathcal{V}_{G}^{r}|+|\mathcal{V}_{S}^{r}|+|\mathcal{E}^{r}||\mathcal{V}_{G}|+|\mathcal{E}^{r}||\mathcal{V}_{S}|)+|\mathcal{E}_{G}|+4|\mathcal{V}_{S}|+4|\mathcal{R}|+2|\mathcal{V}|+5|\mathcal{V_{G}}|+|\mathcal{V}_{D}|)$ constraints
\end{proof}

Next, we consider the run-time complexity of RHCOP.  Recall that it has two parts, i.e., offline training and online execution. Since the training phase is performed offline using historical data, it does not consume real-time computational resources. Thus, we focus only on its online execution part. 
\begin{prop}
    \label{ProO1}
    The run-time complexity of RHCOP per time slot is $\mathcal{O}((|\mathcal{V}| + |\mathcal{E}_{W}|) K + 2^{KQ})$.
\end{prop}
\begin{proof}
    In each slot, RHCOP's online phase includes a prediction phase and a decision phase. In the prediction phase, RHCOP calls the trained GMMs over a time window with a size of $K$, see Lines 17-22. The time complexity of this phase is $\mathcal{O}((|\mathcal{V}| + |\mathcal{E}_{W}|)K)$. 
    Next, in the decision phase, RHCOP solves the formulated MILP over the given time window. 
    Let $Q$ denote the number of integer variables per time slot in the MILP.  As shown in Proposition \ref{Pro1}, the value of $Q$ is $|\mathcal{V}_{D}|+\sum_{r\in\mathcal{R}}|\mathcal{V}_{G}^{r}||\mathcal{V}_{G}|+\sum_{r\in\mathcal{R}}|\mathcal{V}_{S}^{r}||\mathcal{V}_{S}|+\sum_{r\in\mathcal{R}}|\mathcal{V}_{G}^{r}||\mathcal{V}_{S}^{r}||\mathcal{V}_{G}||\mathcal{V}_{S}|+|\mathcal{R}|$.  The worst-case complexity of solving this MILP using standard branch-and-bound methods is $\mathcal{O}(2^{KQ})$.
    Lastly, both of the prediction phase and the decision phase run in each slot.  Hence, the time complexity of the online execution in RHCOP is $\mathcal{O}((|\mathcal{V}| + |\mathcal{E}_{W}|) K + 2^{KQ})$. 
\end{proof}
%
Note that the online execution phase of RHCOP can be sped up by pre-computing solutions offline and storing them in a neural network, see \cite{10870303} for an example work.  This means an operator does not have to solve an MILP in each time slot.   Instead, it retrieves the corresponding solution from the neural network based on energy information  at gateways and servers.
%

\begin{prop}
    \label{ProO}
    The time complexity of GreedyOL is $\mathcal{O}(|\mathcal{T}||\mathcal{R}|(|\mathcal{R}|log|\mathcal{R}|+|\mathcal{V}_{G}^{r}||\mathcal{V}_{G}||\mathcal{N}_{i}|+|\mathcal{V}_{S}^{r}||\mathcal{V}_{S}|+|\mathcal{E}^{r}|))$.
\end{prop}
\begin{proof}
    The Sort function has time complexity $\mathcal{O}(|\mathcal{R}|log|\mathcal{R}|)$.
    In lines 3-8, each VNF-C calls function GatewayECD(.) that runs at most $|\mathcal{N}_{i}|$ times.  Thus, lines 3-8 run at most $|\mathcal{V}_{G}^{r}||\mathcal{V}_{G}||\mathcal{N}_{i}|$ times to select a gateway for VNF-C. 
    Lines 9-14 run at most $|\mathcal{V}_{S}^{r}||\mathcal{V}_{S}|$ times when selecting a server for VNF-P.  
    Lines 15-18 run at most $|\mathcal{E}^{r}|$ times when checking whether link $(\eta_{u}^{r}, \psi_{v}^{r})$ supports edge (u,v) or not.  
    {\em UpdateRe(.)} searches all gateways and servers, and determine whether to update their energy, computing and bandwidth resources.  
    Its time complexity is thus $\mathcal{O}(|\mathcal{V}_{G}|+|\mathcal{V}_{S}|)$.
    Next, we observe that lines 2-25 run at most $|\mathcal{R}|$ times to check all DAGs.  Further, as GreedyOL works over $|\mathcal{T}|$ slots, the time complexity of GreedyOL is $\mathcal{O}(|\mathcal{T}||\mathcal{R}|(|\mathcal{R}|log|\mathcal{R}|+|\mathcal{V}_{G}^{r}||\mathcal{V}_{G}||\mathcal{N}_{i}|+|\mathcal{V}_{S}^{r}||\mathcal{V}_{S}|+|\mathcal{E}^{r}|)$
\end{proof}
\section{Evaluation}\label{EVAL}
We conduct our simulation using Python 3.8.5 and Gurobi 9.9.1.  Gateways, devices, servers and sink $s$ are randomly deployed on a $1000\times1000$ $m^{2}$ area.   Devices have a solar panel of size $30\times30$ ${cm}^{2}$ with efficiency 20\%.  
Tables \ref{HMM1} and \ref{HMM2} list parameter values of the solar energy arrival process \cite{Data_Driven_Solar}.
As per \cite{Data_Driven_Solar}, the mean of four solar energy arrival states are 94.6, 76.0, 45.6 and 17.9 $mW/cm^2$.  Further, their respective variance is 0.31, 1.55, 1.48 and 0.71.
The battery size of a device, a gateway and a server is set to 10 J, 100 J and 100 J, respectively. 
Gateways and servers have a CPU capacity of 1000 Megacycles~\cite{NFV_Energy_2020_10}.  
Their base and peak power are set to 170 W and 500 W, respectively \cite{NFV_Energy_2020_10}.
We set $\rho$ to 150 nJ/bit.          

We set $B$ and $N_{0}$ to 200 kHz and -95 dBm/Hz, respectively.  The path loss exponent $\alpha$ is 2.5 and the path loss at the reference distance of one meter is 30 dB.  
Wired links have a capacity of 1000 Mbps \cite{NFV_Energy_2020_10}.   
Each DAG request has five VNFs that includes at least one VNF-C and one VNF-P.  Other VNF types are generated randomly.
The computation resource of each VNF is drawn uniformly from the range [10,100] (in Megacycles).  Each VNF-C randomly connects to a VNF-P with a probability of 0.9. 
The bandwidth demand between VNFs is randomly drawn from [10, 50] (in kb/s).  
The data requirement of VNF-C is 100 kb/s, and we set $\Psi = |\mathcal{T}|$.

\begin{table}[htbp]
\renewcommand{\arraystretch}{1.2}
\caption{Mean and variance of the Solar Energy Arrival}\label{HMM1}
\centering
\begin{tabular*}{\linewidth}{@{\extracolsep{\fill}} c c c c c @{}} 
\hline
    State & Poor & Fair & Good & Excellent \\ 
\hline
    Mean $(mW/cm^{2})$ & 1.75 & 4.21 & 7.02 & 9.38 \\ 
\hline
    Variance & 0.65 & 1.04 & 2.34 & 0.54\\
\hline
\end{tabular*}
\end{table}
\begin{table}[htbp] 
\renewcommand{\arraystretch}{1.5}
\caption{Transition Probability of the Solar Energy Arrival} \label{HMM2}
\centering
\begin{tabular*}{\linewidth}{@{\extracolsep{\fill}} c c c c c @{}}
\hline
State & Poor & Fair & Good & Excellent \\ \hline
Poor & 0.979 & 0.015 & 0.006 & 0 \\ \hline
Fair & 0.005 & 0.988 & 0.007 & 0 \\ \hline
Good & 0.006 & 0.009 & 0.975 & 0.010 \\ \hline
Excellent & 0 & 0 & 0.007 & 0.993 \\
\hline
\end{tabular*}
\end{table}

\subsection{Benchmark Solution}
We benchmark against two solutions called Random and GMMPre.
Random selects a random number of DAGs to embed in each slot, and then runs MILP.  If the MILP is solvable, Random embeds the DAGs and updates the computing, bandwidth and energy resources of devices, gateways and servers, respectively.  
Otherwise, it does not embed the selected DAGs, and moves to the next slot. 
GMMPre employs GMM to estimate the energy arrivals of $K^{'}=8$ future time slots.  Next, it calculates the average energy harvesting rate of each device, gateway and server.  Lastly, it uses this average value to replace the energy level in GreedyOL, and then calls GreedyOL.
We study eight parameters, i.e., $|\mathcal{V}_{G}|$, $|\mathcal{V}_{S}|$, $|\mathcal{N}_{i}|$, $|\mathcal{R}|$, $K$, $L$, VNF-C/VNF-P in a DAG, and number of VNFs.
We report the average result of 50 runs.

\begin{figure*}
	\centering
	\subfloat[] 
	{\includegraphics[width=.25\textwidth] {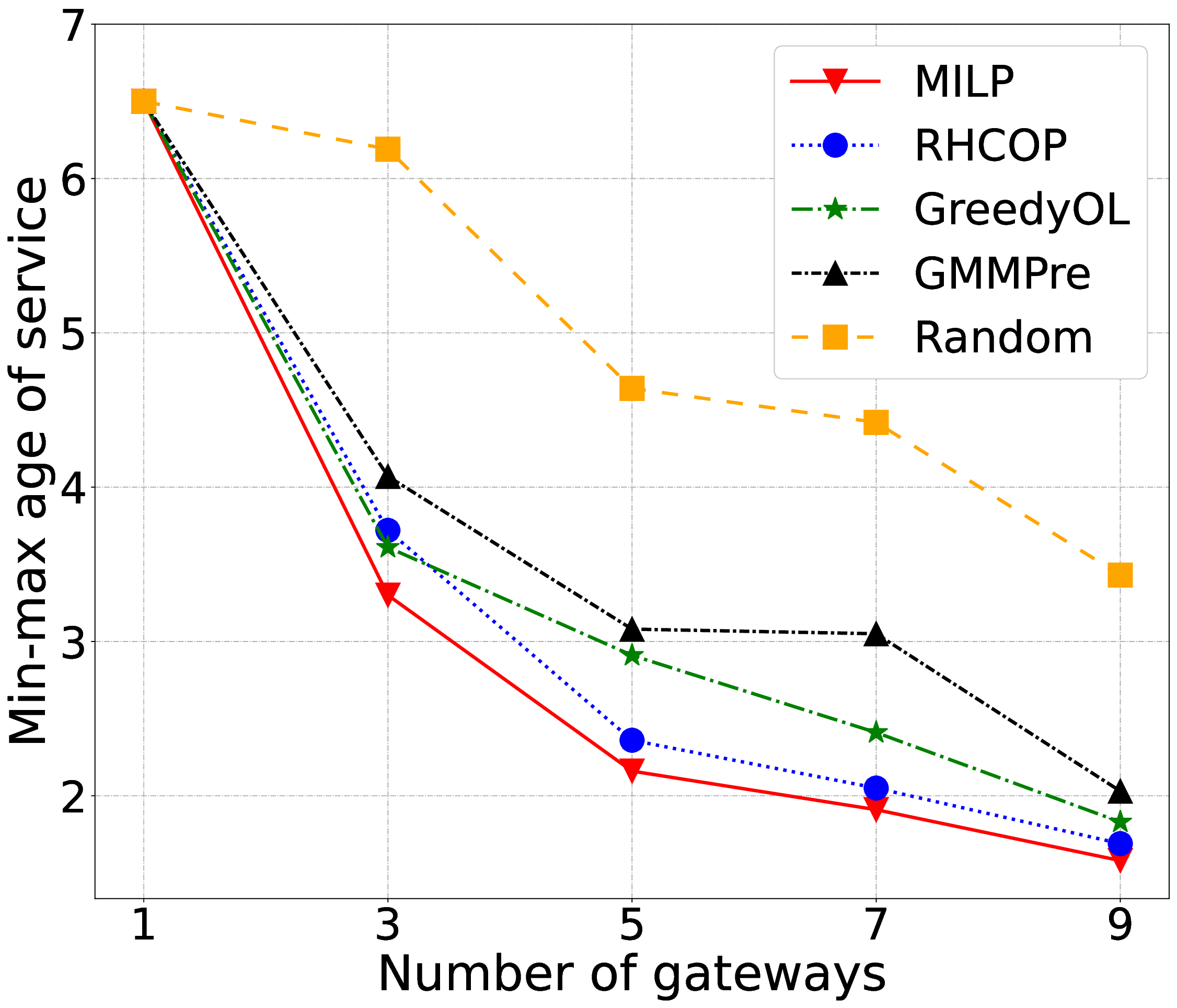}\label{gateway}}\hfill
	\subfloat[] 
	{\includegraphics[width=.25\textwidth] {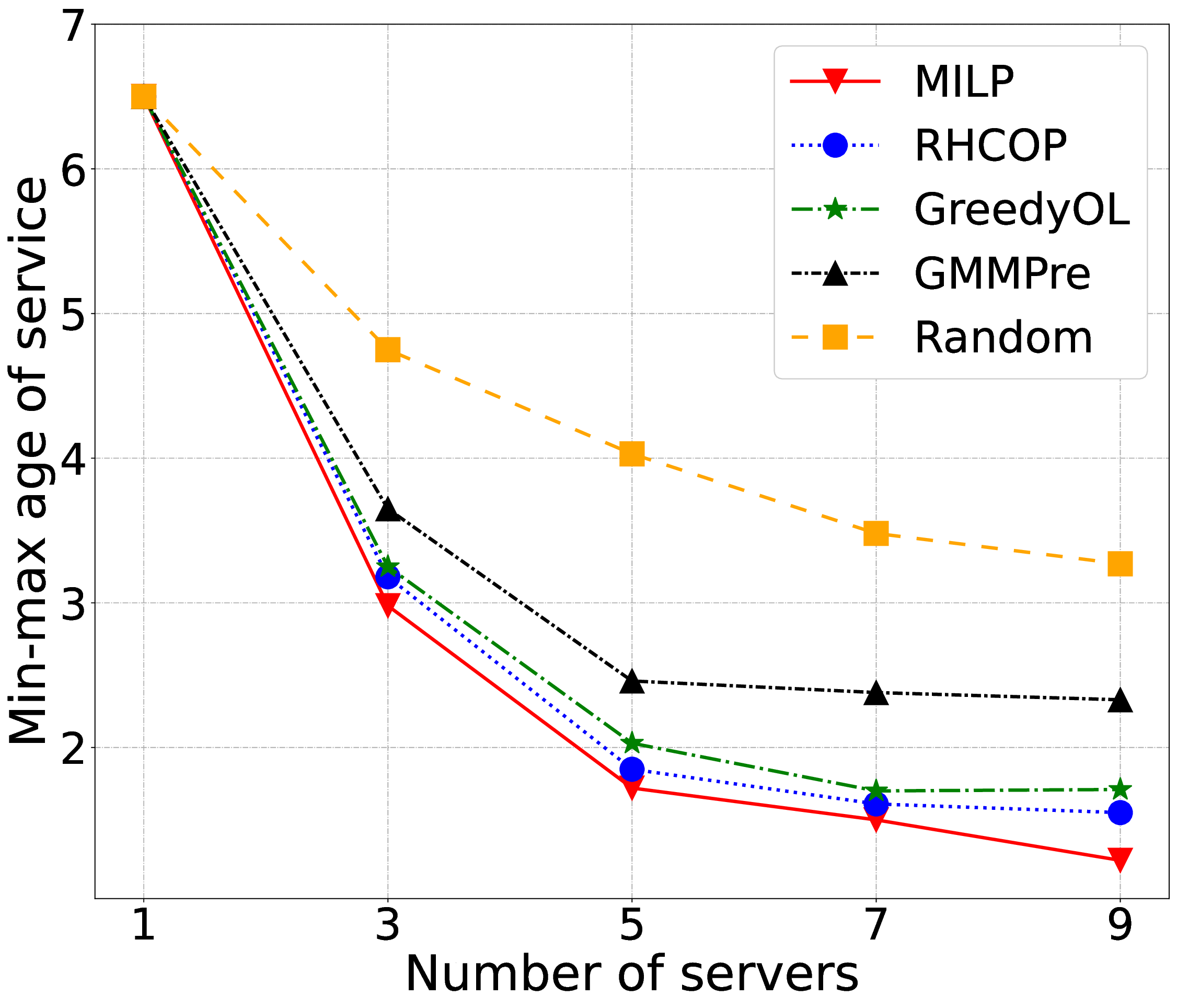}\label{server}}\hfill
	\subfloat[] 
	{\includegraphics[width=.25\textwidth] {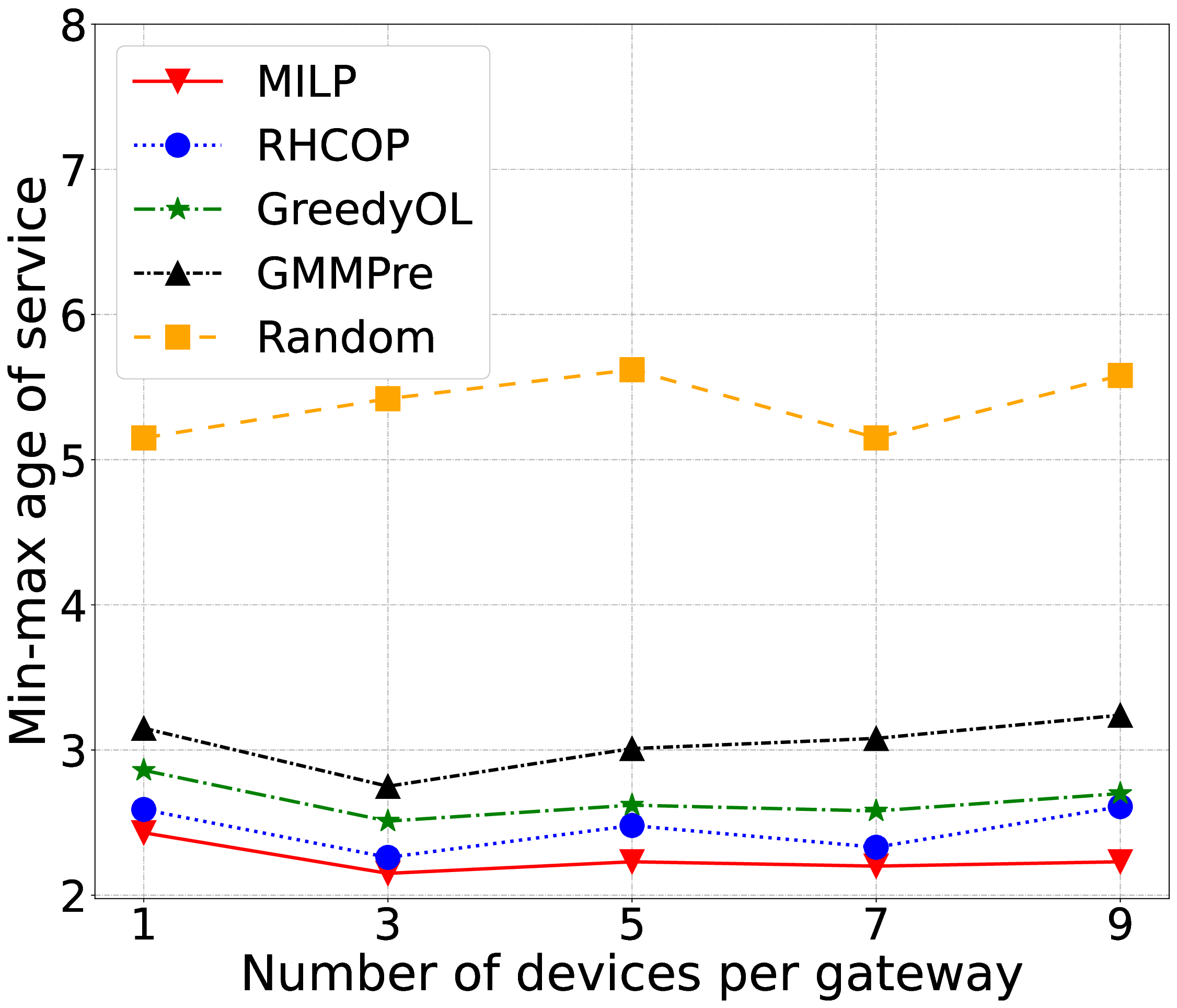}\label{DeviceFig}}\hfill
	\subfloat[] 
	{\includegraphics[width=.25\textwidth] {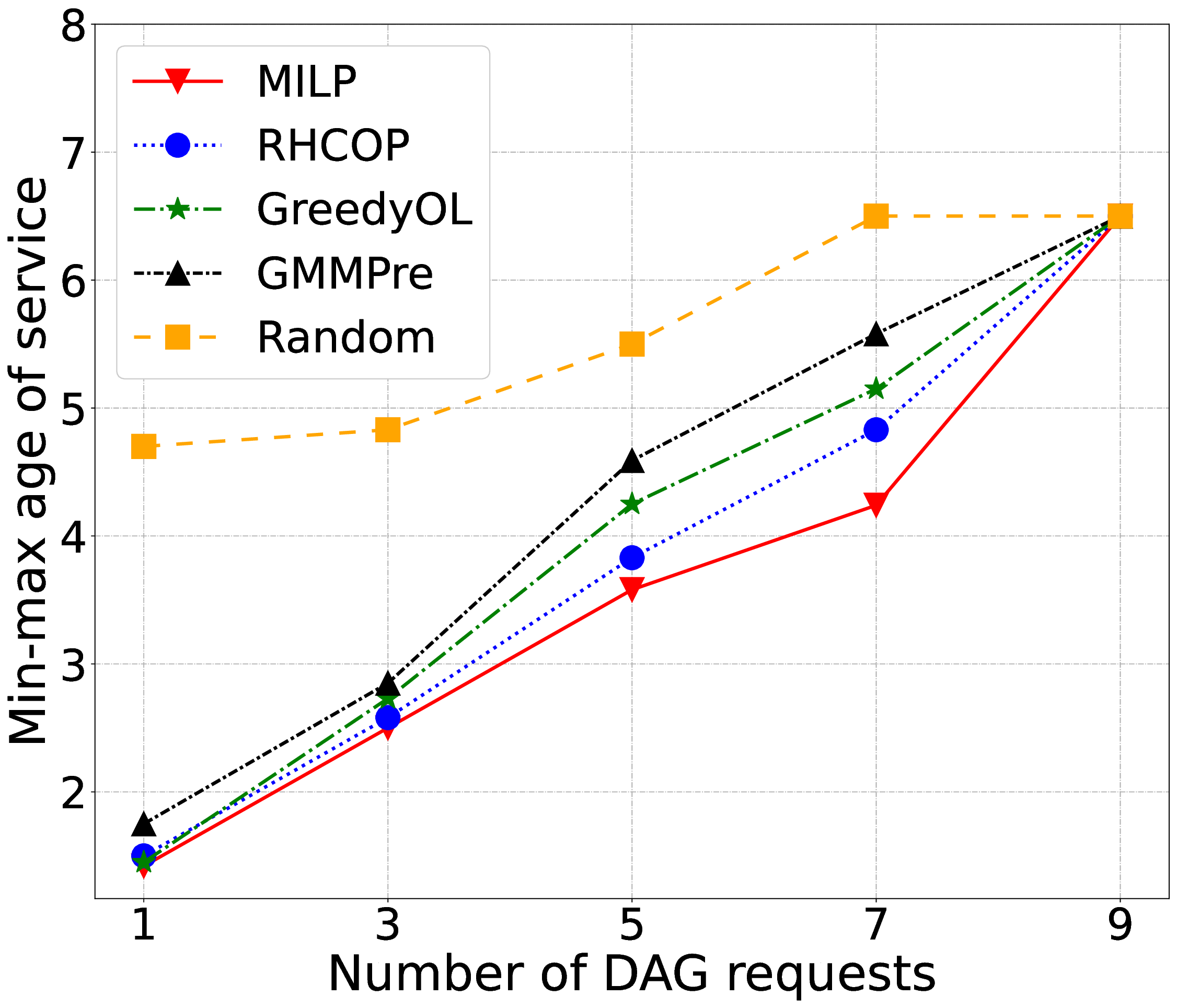}\label{DAG}}\hfill
	\subfloat[] 
	{\includegraphics[width=.25\textwidth] {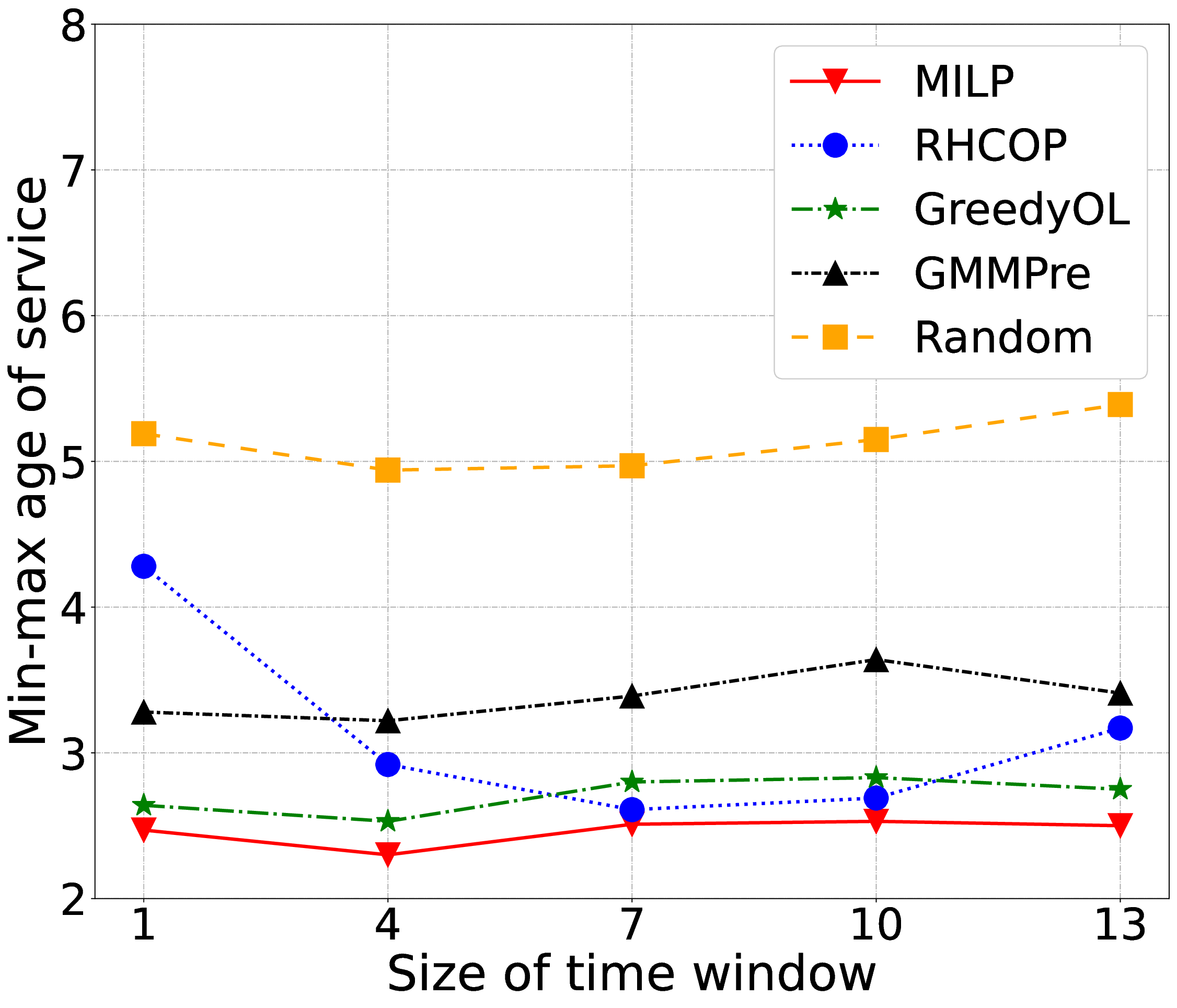}\label{Window}}\hfill
	\subfloat[] 
	{\includegraphics[width=.25\textwidth] {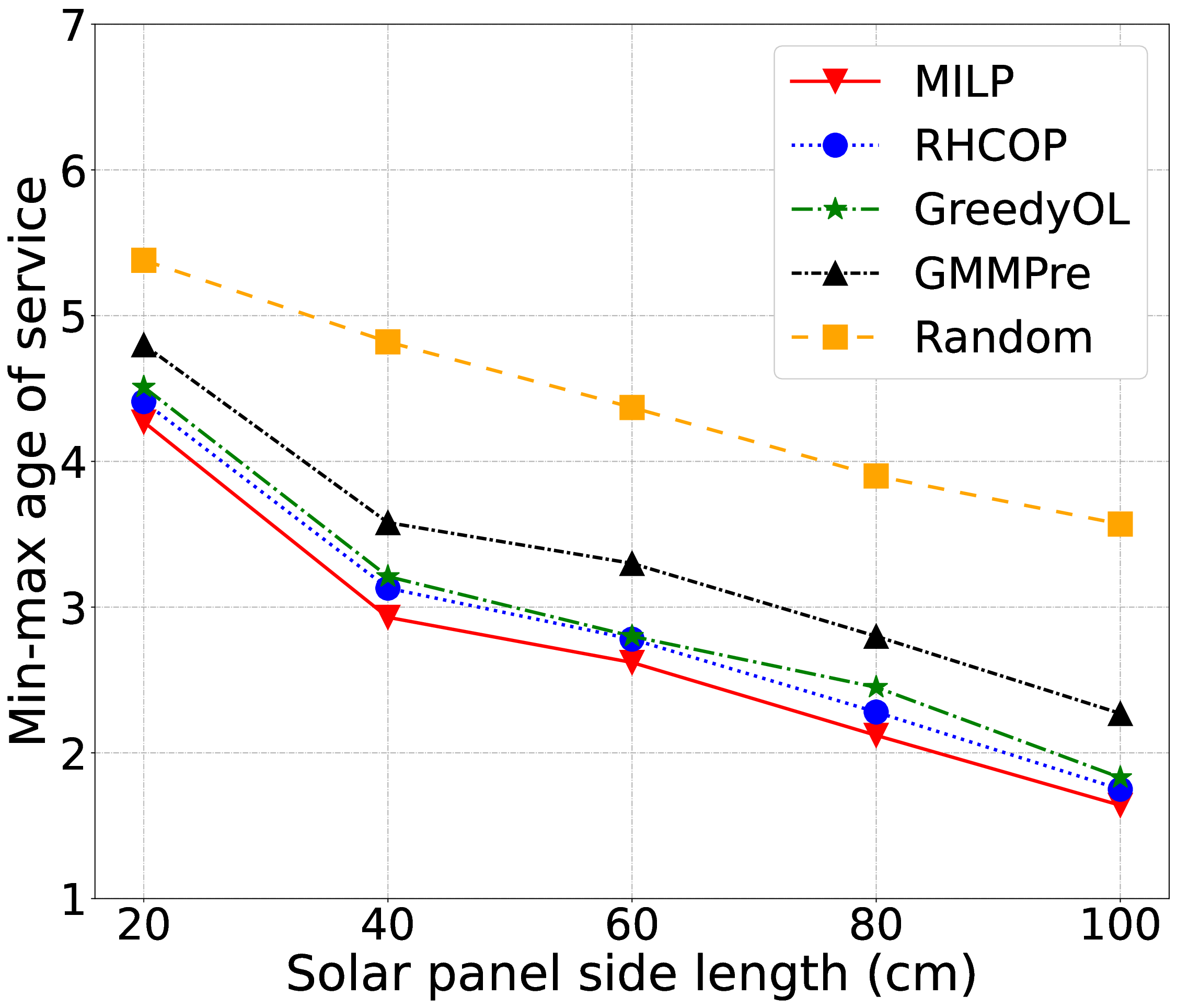}\label{Solar_P}}\hfill
    \subfloat[] 
	{\includegraphics[width=.25\textwidth] {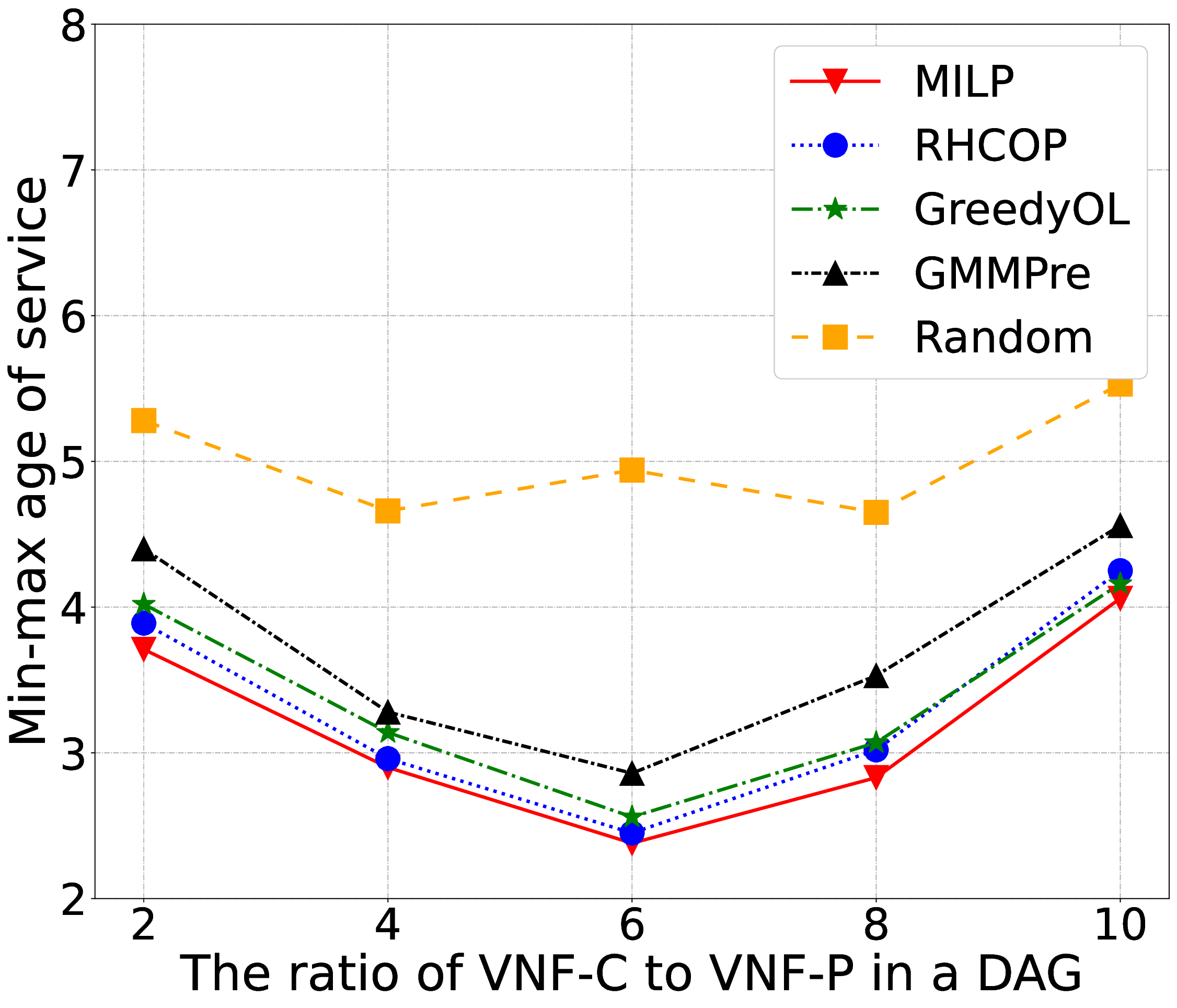}\label{VNFRatio}}\hfill
    \subfloat[] 
	{\includegraphics[width=.25\textwidth] {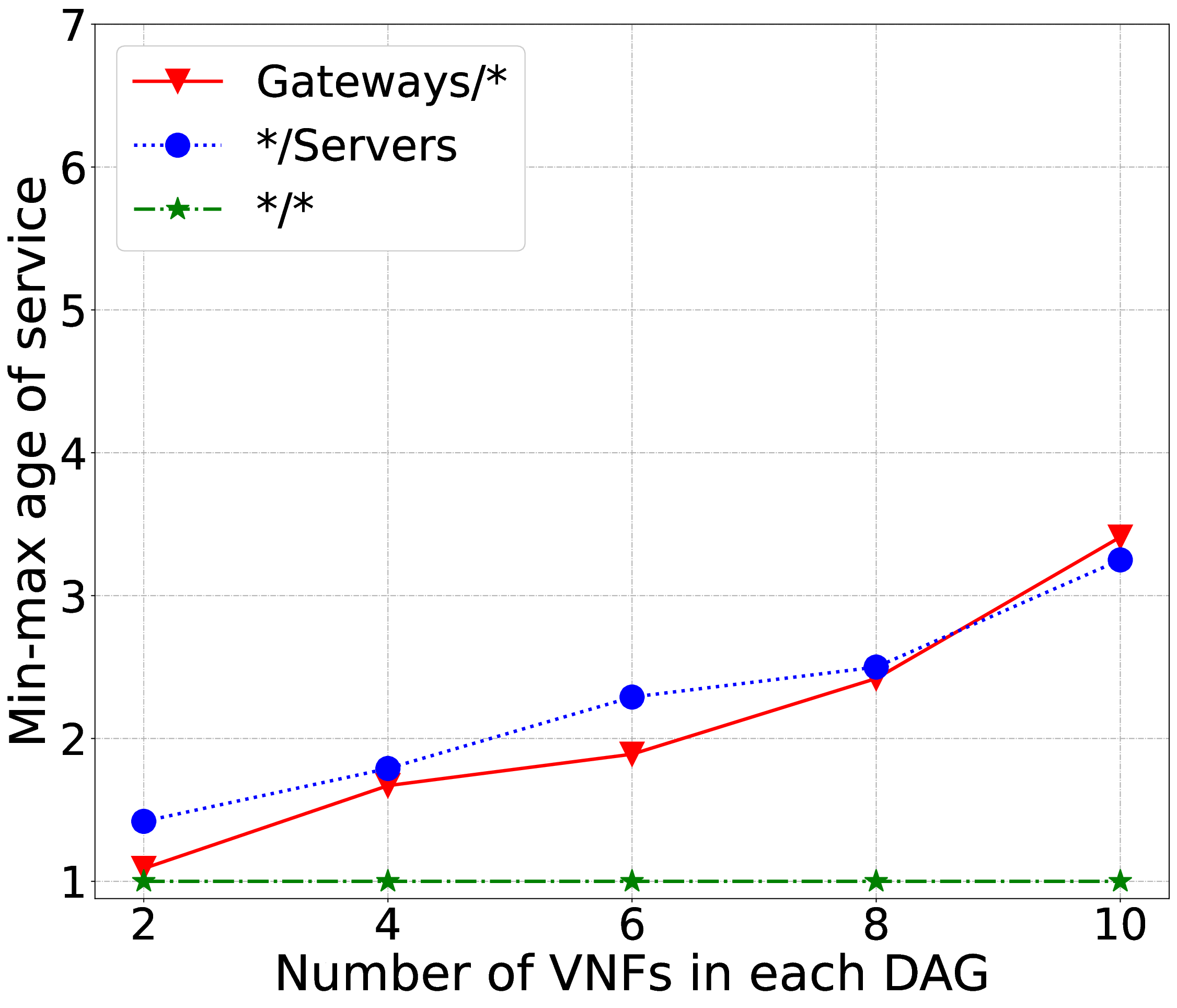}\label{EHCom}}\hfill
	\caption{Impact of various parameters on min-max AoS: (a) number of gateways $|\mathcal{V}_{G}|$, (b) number of servers $|\mathcal{V}_{S}|$, (c) number of devices associated to a gateway $|\mathcal{N}_{i}|$, (d) number of DAG requests $|\mathcal{R}|$, (e) window size $K$, (f) panel size of servers and gateways $L$, (g) the ratio of VNF-C to VNF-P in a DAG, (h) infinite resource.}\label{Small_Net}  
\end{figure*}
\subsection{Simulation Results}
We first study how the number of gateways, i.e., $|\mathcal{V}_{G}|$, impact max-min AoS.
We set $|\mathcal{V_{S}}|$, $|\mathcal{N}_{i}|$, $|\mathcal{R}|$, $K$ and $L$ to 3, 3, 3, 8 and 30 cm, respectively.  There are 12 slots, i.e., $|\mathcal{T}|=12$.  
As shown in Figure \ref{gateway}, min-max AoS decreases with $|\mathcal{V}_{G}|$. 
For example, the min-max AoS of MILP and RHCOP decreases from 6.50 and 6.50 to 1.58 and 1.69, respectively.  The value of GreedyOL, GMMPre and Random decreases from 6.50 to 1.83, 2.03 and 3.43, respectively.  
The reason is because there are more potential gateways to deploy the VNF-Cs of a DAG in each slot when $|\mathcal{V}_{G}|$ increases.  
Thus, more computation, bandwidth and energy resources can be used to support a VNF-C in each slot, which reduces AoS. 
Another observation is that the min-max AoS of all solutions is bounded to 6.50 when $|\mathcal{V}_{G}|=1$.  
This reason is because one gateway has insufficient energy, computing and bandwidth resource to service DAG requests.
All DAG requests cannot be active during $|\mathcal{T}|$ slots, causing $a_{s}^{t}$ to increase by one in each slot.  As the AoS of each DAG request over $|\mathcal{T}|=12$ frames is $1+2+\dots+|\mathcal{T}|$, its average value is 6.5.
Next, we vary the number of servers, where $|\mathcal{V}_{S}| \in \{1,3,5,7,9\}$. 
The value of $|\mathcal{V_{G}}|$, $|\mathcal{N}_{i}|$, $|\mathcal{R}|$ is set to three.  The value of $K$, $L$ and $|\mathcal{T}|$ is 8, 30 cm and 12 slots, respectively.
Referring to Figure \ref{server}, the min-max AoS of MILP decrease from 6.50 to 1.22 as $|\mathcal{V}_{S}|$ increases from one to nine. RHCOP shows a reduction from 6.50 to 1.55, while GreedyOL decreases from 6.50 to 1.71.  The value of GMMPre and  Random decreases by 64.15\% and 49.69\% from 6.50, respectively.  
The reasons are as follows.  First, more servers provide more energy, communication and computing resources for VNF-Ps in DAG requests. 
Further, a higher $|\mathcal{V}_{S}|$ value leads to more paths from gateways to the sink, where a DAG request selects a path that consumes less energy to deliver sensor data.  This then allows a DAG to activate more times over $|\mathcal{T}|$ slots in higher $|\mathcal{V}_{S}|$ case.
Next, we consider the number of devices associated to each gateway, i.e., $|\mathcal{N}_{i}|$.  
We set $|\mathcal{V_{G}}|=|\mathcal{V}_{S}|=|\mathcal{R}|=3$, $K=8$, $L=30$ cm and $|\mathcal{T}|=12$ slots.
As per Figure \ref{DeviceFig}, the number of devices associated at each gateway has limited impact on min-max AoS.  
For example, the min-max AoS of MILP ranges between 2.43 to 2.23 when $|\mathcal{N}_{i}|$ increases from one to nine.  For RHCOP, it ranges between 2.59 to 2.61 whereas for GreedyOL it ranges from 2.86 to 2.70.  
The reason is because devices harvest sufficient energy for sensing and data upload.  As a result, the number of devices associated to each gateway has limited impact on the min-max AoS.     
Next, we consider the number of DAG requests, where $|\mathcal{R}|\in\{1,3,5,7,9\}$. 
We set $|\mathcal{V_{S}}|$, $|\mathcal{V_{G}}|$, $|\mathcal{N}_{i}|=3$.  The value of $K$, $L$ and $|\mathcal{T}|$ is 8, 30 cm and 12 slots, respectively.  
As shown in Figure \ref{DAG}, the min-max AoS of MILP increased by 357.75\% from 1.42 when $|\mathcal{R}|$ increased from one to nine.  
The min-max AoS of RHCOP, GreedyOL, GMMPre and Random increases by 333.33\%, 348.28\%, 271.43\% and 38.30\%, where these values are 6.50 when $|\mathcal{R}|=9$. 
This is because higher number of DAG requests require more energy, communication and computation resources to run all DAG requests.  As these resources are fixed, the min-max AoS of all solutions increases.  
The next experiment studies the following window sizes: $K\in \{1, 4, 7, 10, 13\}$, where $|\mathcal{V_{S}}| =|\mathcal{V_{G}}|=|\mathcal{N}_{i}|=|\mathcal{R}|=3$.  The value of $L$ and $|\mathcal{T}|$ is set to 30 cm and 12 slots, respectively. 
From Figure \ref{Window}, the time window size has no impact on the performance of MILP, GreedyOL, GMMPre and Random.  For example, the min-max AoS of MILP ranges from 2.47 to 2.50 when $K$ increases from one to 13.  
This is because MILP, GreedyOL and Random do not use a time horizon window.  Further, the window size of GMMPre is fixed to $K^{'}=8$.  
Consequently, the value of $K$ has no impact on the performance of GMMPre. 
Another observation is that the min-max AoS of RHCOP decreases from 4.28 to 2.61 when $K$ increases from one to seven. 
The reason is because RHCOP has more wireless channel gain and energy arrival information of future slots when $K$ increases.
This thus allows RHCOP to make decision using more prediction information.
However, due to estimation errors, with more future slots, a solution computed by RHCOP for the current time slot is less likely to be optimal.   Thus, the min-max AoS of RHCOP increases from 2.61 to 3.17 when $K$ further increases to 13.  
Next, we study the impact of solar panel size $L$, where $L$ increases from 20 cm to 100 cm with an interval of 20 cm.   We set $|\mathcal{V_{S}}| =|\mathcal{V_{G}}|=|\mathcal{N}_{i}|=|\mathcal{R}|=3$, $K$=8 and $|\mathcal{T}|=12$. 
Referring to Figure \ref{Solar_P}, the min-max AoS of MILP decreased by 61.59\% from 4.27 when $L$ increased from 20 cm to 100 cm.   The min-max AoS of RHCOP and GreedyOL is 4.41 and 4.51 when $L=20$ cm, respectively.  These values decrease to 1.75 and 1.83 when $L$ further increases to 100 cm.   
This is because gateways and servers have more energy to support VNF-Cs and VNF-Ps, leading to lower AoS values.  
We now study DAGs with different ratio of VNF-Cs versus VNF-Ps.  
We fix the number of VNFs in a DAG to six, and study the following ratio of VNF-C to VNF-P in each DAG: \{1/6, 2/6, 3/6, 4/6, 5/6\}.  We set $|\mathcal{V_{S}}| =|\mathcal{V_{G}}|=|\mathcal{N}_{i}|=|\mathcal{R}|=3$, where $K$ and $L$ is 8 and 30 cm, respectively.
As per Figure \ref{VNFRatio}, min-max AoS decreases when the ratio of VNF-C to VNF-P increases from 1/6 to 3/6.  For example, the min-max AoS of MILP increases from 3.71 to 2.38.  However, as the ratio of VNF-C to VNF-P further increases from 3/6 to 5/6, the min-max AoS increases.  For instance, the min-max AoS of MILP increases from 2.38 to 4.06.  
The reason is because the increasing ratio of VNF-Cs to VNF-Ps resulting in fewer VNF-Ps.  Hence, fewer VNF-Ps share energy, communication and computation resources of servers, leading to lower min-max AoS values.  
However, as the ratio of VNF-Cs to VNF-Ps further increases, more VNF-Cs share the resources of gateways, resulting in higher min-max AoS values.  
This experiment studies how resource limit at gateways and servers impact performance, where we study three scenarios: (i) gateways with unlimited resources, (ii) servers with unlimited resources, (iii) servers and gateways with unlimited resources.  
We set $|\mathcal{V_{S}}| =|\mathcal{V_{G}}|=|\mathcal{N}_{i}|=|\mathcal{R}|=3$.  The value of $L$, $K$ and $|\mathcal{T}|$ is 30 cm, 8, and 12 slots, respectively. 
The number of VNFs in a DAG increases from two to ten with an interval of two.  
Referring to Figure \ref{EHCom}, the min-max AoS of scenario (iii) is one, meaning the network has sufficient resources to activate all DAGs in each slot.  
The min-max AoS of scenario (i) and (ii) increases by 212.84\% and 128.87\% from 1.09 and 1.42, respectively, with increasing number of VNFs.  This is expected as more resources are required to run VNFs.  
In scenario (i), more VNF-Ps share the resource of servers when the number of VNFs increases.  Thus, the min-max AoS of MILP increases.  When servers have infinite resource in scenario (ii) additional VNFs mean more VNF-Cs will require the resource of gateways, which increases the min-max AoS of MILP.  

\begin{figure}
    \centering
    \includegraphics[width=0.8\linewidth]{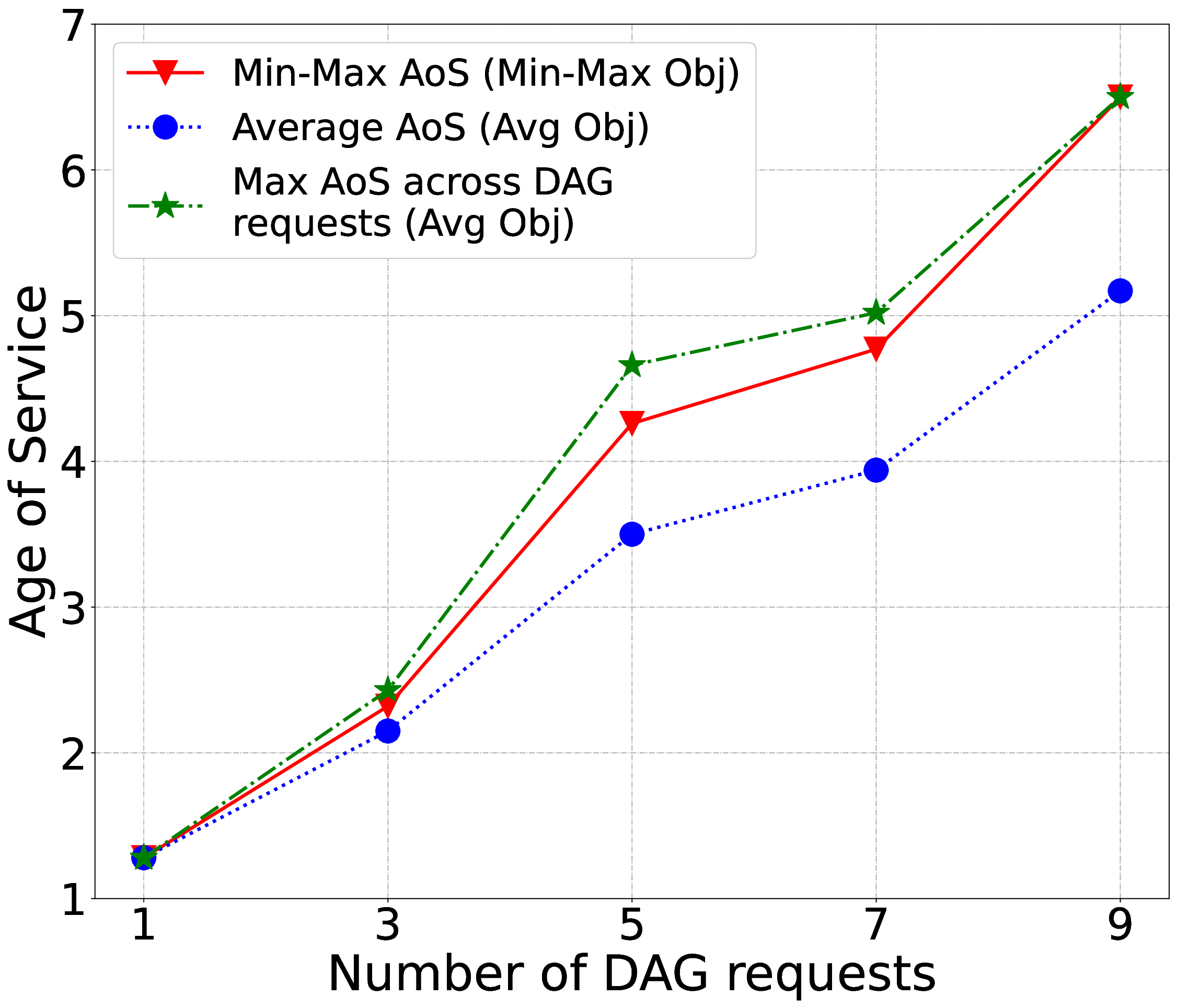}
    \caption{Min-max AoS vs. average AoS.}
    \label{fig:avg}
\end{figure}

We now discuss optimizing min-max AoS versus average AoS, and the trade-off between fairness and network efficiency.  
We set $|\mathcal{V_{S}}| =|\mathcal{V_{G}}|=|\mathcal{N}_{i}|=3$, $K$=8 and $|\mathcal{T}|=12$, and increase $|\mathcal{R}|$ from one to nine.  We compare three schemes: (i) the min-max AoS achieved by MILP \ref{milp}, (ii) the average AoS, which the the average AoS of all DAG requests, (iii) the maximum average AoS.
As illustrated in Figure \ref{fig:avg}, the average AoS is lower than min-max AoS.  For example, the average AoS is 82.16\% that of the min-max AoS for $|\mathcal{R}|=5$.  This means our approach achieves higher performance when we optimize the average AoS.  
However, the min-max AoS is lower than the maximum value of average AoS.  For example, the value of min-max AoS is 91.42\% that of the maximum value of average AoS when $|\mathcal{R}|=5$.  
It shows that relying solely on the average AoS leads to significantly higher AoS for certain nodes, resulting in poor fairness.  In contrast, our min-max AoS metric ensures fairness across all nodes.

\begin{figure}
    \centering
    \includegraphics[width=0.85\linewidth]{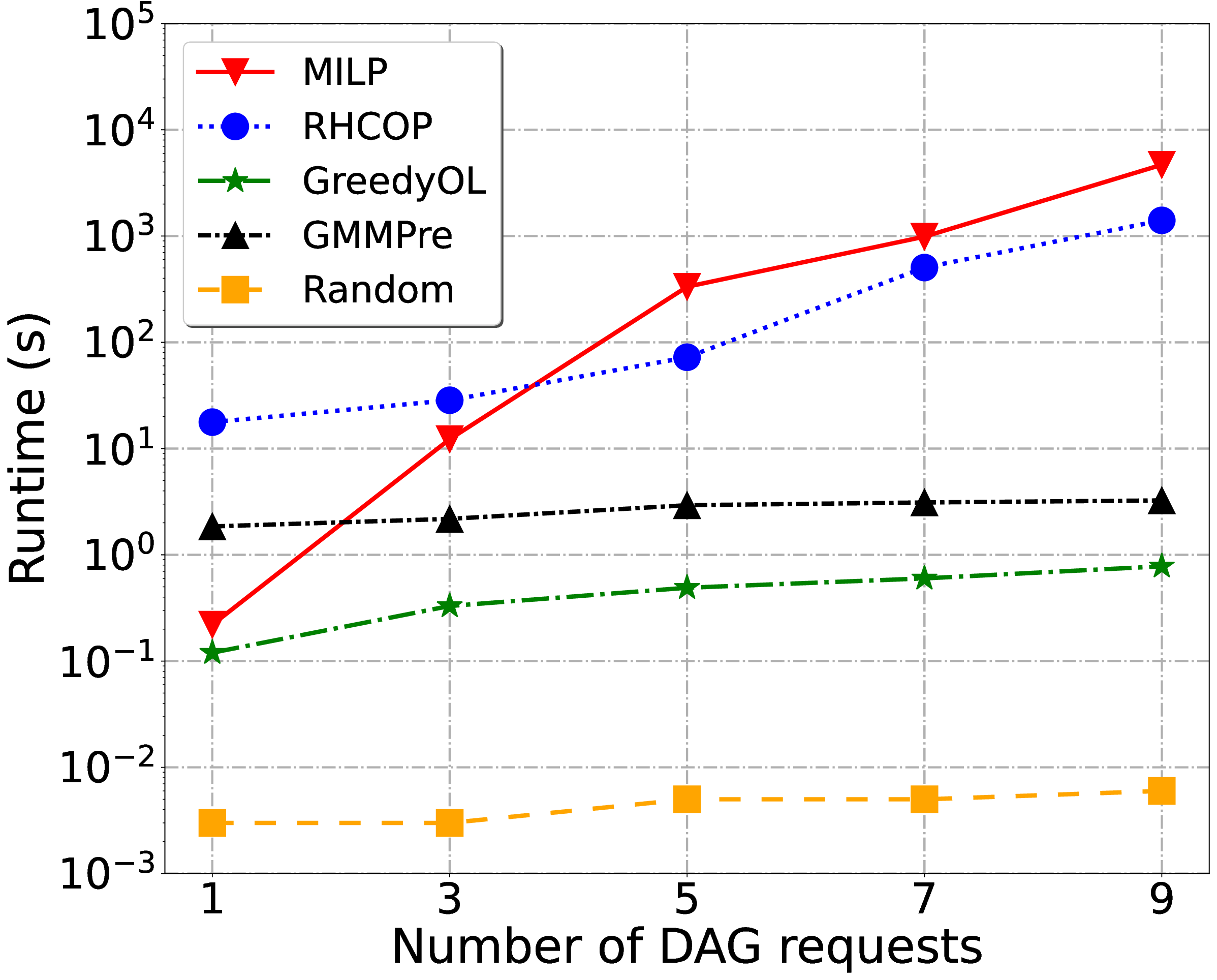}
    \caption{Runtime of MILP, RHCOP, GreedyOL, GMMPre and Random. }
    \label{fig:runtime}
\end{figure}

We now analyze the runtime of MILP, RHCOP, GreedyOL, GMMPre and Random as the number of DAG requests varies, where $|\mathcal{R}|\in\{1,3,5,7,9\}$.
We set $|\mathcal{V_{S}}| =|\mathcal{V_{G}}|=|\mathcal{N}_{i}|=3$, $K$=8 and $|\mathcal{T}|=12$.
Referring to Figure~\ref{fig:runtime}, the runtime of MILP increases as $|\mathcal{R}|$ grows, rising from 0.22 seconds at $|\mathcal{R}|=1$ to 4685.04 seconds at $|\mathcal{R}|=9$.  RHCOP increases from 17.73 seconds to 1403.92 seconds over the same range. 
This rapid growth occurs because both MILP and RHCOP rely on a solver to solve MILP, where the number of decision variables and constraints increases with $|\mathcal{R}|$.
For GreedyOL, the runtime increases from 0.12 seconds to 0.78 seconds when $|\mathcal{R}|$ increases one to nine, which empirically validates the complexity analysis in Section~\ref{ANA}.
Lastly, we consider the min-max AoS gap between MILP, RHCOP, GreedyOL, GMMPre and Random.
In particular, compared to the min-max AoS of MILP, the performance of RHCOP, GreedyOL, GMMPre and Random is 1.07, 1.13, 1.29 and 1.92 higher than that of MILP.
The reason is because MILP uses non-causal information to make decision, i.e., future wireless channel gains and energy arrivals over $|\mathcal{T}|$ slots.  By contrast, RHCOP, GreedyOL, GMMPre and Random have only causal information.
Further, RHCOP and GMMPre use GMM to predict wireless channel gains and energy arrivals,  resulting in prediction error.  
In addition, we observe that the performance of RHCOP is better than that of GreedyOL, GMMPre and Random.  This is because (i) RHCOP uses historical data, and (ii) RHCOP makes decision using estimates across multiple slots.
GreedyOL has better performance as compared to GMMPre.  The reason is because GMMPre uses the average energy arrivals at gateways or servers, meaning in some time slots, a gateway or server may have insufficient energy to support the decision of GMMPre.

%
%
\section{Conclusion}\label{CONC}
This paper has considered a novel DAG requests embedding problem in solar powered networks.  Its aim is to minimize the maximum AoS of DAGs/applications.  It presented the first MILP for the problem, and proposed a RHC-based solution called RHCOP, and a heuristic solution called GreedyOL.
The results showed that the min-max AoS of RHCOP and GreedyOL is 1.07x and 1.13x higher than MILP.  
Further, more gateways, servers and larger solar panel sizes helped reduce min-max AoS.  By contrast, increasing the number of DAG requests led to higher min-max AoS values.  
All solutions had the lowest min-max AoS when the ratio of VNF-C to VNF-P is 0.5.
A potential future work is to consider applications that require specific services or data.  
\bibliographystyle{ieeetr}
\bibliography{r}

\end{document}